\documentclass[10pt]{article}
\usepackage{graphicx,epsfig,amsmath,amsthm,amssymb}
\usepackage{enumerate}
\usepackage{multirow}
\usepackage{bbm}
\usepackage{balance}
\usepackage{footnote}
\usepackage{tablefootnote}
\usepackage{xr}
\usepackage{hyperref}
\newcommand{\ex}[1]{\mathbb{E}\left[ #1 \right] }
\newcommand{\norm}[1]{\left\lVert #1 \right\rVert}
\newcommand{\Q}{ {\mathbf Q}}

\newcommand{\Qc}{ {\mathbf Q}_{\perp}}
\newcommand{\Qp}{ {\mathbf Q}_{\parallel}}
\newcommand{\A}{ {\mathbf A}}
\newcommand{\s}{ {\mathbf S}}
\newcommand{\UU}{ {\mathbf U}}
\newcommand{\inner}[2]{\langle #1, #2 \rangle}
\newcommand{\lep}[1]{\mathop  \le \limits^{(#1)}}

\newcommand{\ep}[1]{\mathop  = \limits^{(#1)}}
\newcommand{\norms}[1]{\big\lVert #1 \big\rVert}

\DeclareMathOperator*{\argmin}{arg\,min}

\newtheorem{definition}{Definition}
\newtheorem{lemma}{Lemma}
\newtheorem{claim}{Claim}
\newtheorem{proposition}{Proposition}

\newtheorem{theorem}{Theorem}

\newtheorem{example}{Example}
\newtheorem{remark}{Remark}

\theoremstyle{definition}

\begin{document}






%

\title{Asymptotically Optimal Load Balancing in Large-scale Heterogeneous Systems with Multiple Dispatchers}
\author{Xingyu Zhou\\Department of ECE\\The Ohio State University\\zhou.2055@osu.edu 
\and Ness Shroff\\Department of ECE and CSE\\The Ohio State University\\
shroff.11@osu.edu \and Adam Wierman\\Department of Computing and Mathematical Sciences\\California Institute of Technology\\adamw@caltech.edu\\} 
\date{}
\maketitle

\begin{abstract}

We consider the load balancing problem in large-scale heterogeneous systems with multiple dispatchers. We introduce a general framework called Local-Estimation-Driven (LED). Under this framework, each dispatcher keeps local (possibly outdated) estimates of queue lengths for all the servers, and the dispatching decision is made purely based on these local estimates. The local estimates are updated via infrequent communications between dispatchers and servers. We derive sufficient conditions for LED policies to achieve throughput optimality and delay optimality in heavy-traffic, respectively. These conditions directly imply delay optimality for many previous local-memory based policies in heavy traffic. Moreover, the results enable us to design new delay optimal policies for heterogeneous systems with multiple dispatchers. Finally, the heavy-traffic delay optimality of the LED framework directly resolves a recent open problem on how to design optimal load balancing schemes using delayed information.

\end{abstract}


\maketitle


\section{Introduction}
Load balancing, which is responsible for dispatching jobs on parallel servers, has attracted significant interest in recent years. This is motivated by the challenges associated with efficiently dispatching jobs in large-scale data centers and cloud applications, which are rapidly increasing in size. A good load balancing policy not only ensures high throughput by maximizing server utilization, but improves the user experience by minimizing delay.

There have been numerous load balancing policies proposed in the literature. The most straightforward one is Join-Shortest-Queue (JSQ), which has been shown to enjoy optimal delay in both non-asymptotic (for homogeneous servers) and asymptotic regimes~\cite{weber1978optimal,foschini1978basic,eryilmaz2012asymptotically}. However, it is difficult to implement in today's large-scale data centers due to the large message overhead between the dispatcher and servers. As a result, alternative load balancing policies with low message overhead have been proposed. For example, the Power-of-$d$ policy~\cite{mitzenmacher2001power} has been shown to achieve optimal average delay in heavy traffic with only $2d$ messages per arrival~\cite{maguluri2014heavy}. Another common load balancing policy is the pull-based Join-Idle-Queue (JIQ)~\cite{lu2011join,stolyar2015pull}, which has been shown to outperform the Power-of-$d$ policy using less overhead. However, both Power-of-$d$ and JIQ mainly achieve good performance for systems with homogeneous servers. 
Recently, some works consider heterogeneous servers and propose flexible and low message overhead policies that achieve optimal delay in heavy traffic~\cite{zhou2017designing,zhou2018heavy}. However, only a single dispatcher is considered in these works. 
Theoretical analysis of load balancing with multiple dispatchers has mainly focused on the JIQ policy so far~\cite{mitzenmacher2016analyzing,stolyar2017pull}, which has a poor performance in heavy traffic and is even generally unstable for heterogeneous systems~\cite{zhou2017designing}.

Note that heterogeneous systems with multiple dispatchers are now almost the default scenarios in today's cloud infrastructures. On one hand, the heterogeneity comes from the usage of multiple generations of CPUs and various types of devices~\cite{govindan2016evolve}. On the other hand, with the massive amount of data, a scalable cloud infrastructure needs multiple dispatchers to increase both throughput and robustness~\cite{shuff2016building}. 

Motivated by this, a recent work~\cite{anonymous2018load} proposes a new framework named Loosely-Shortest-Queue (LSQ) for designing load balancing policies for heterogeneous systems with multiple dispatchers. In particular, under this framework, each dispatcher keeps its own, local, and possibly outdated view of each server's queue length. Upon arrival, each dispatcher routes to the server with shortest local view. A small amount of message overhead is used to update the local view. The authors successfully establish sufficient conditions on the update scheme for the system to be stable. Moreover, extensive simulations were conducted to show that LSQ policies significantly outperform well-known low-communication policies while using similar communication overhead in both heterogeneous and homogeneous cases. However, no theoretical guarantees on the delay performance are provided. It is worth noting that the key challenge for establishing a delay performance guarantee for this framework is that it only uses possibly outdated local information to dispatch jobs. In fact, the problem of designing delay optimal load balancing schemes that only have access to delayed information has recently been listed as an open problem in~\cite{lipshutz2019open}.

Inspired by this, in this paper, we are particularly interested in the following questions: \emph{Is is possible to establish delay performance guarantees for load balancing in heterogeneous systems with multiple dispatchers? If so, can these guarantees be achieved using only delayed information?} 

\textbf{Contributions.} To answer the questions above, we propose a general framework of load balancing for heterogeneous systems with multiple dispatchers that uses only delayed (out-of-date) information about the system state. We call this framework Local-Estimation-Driven (LED) and it generalizes the LSQ framework. Our main results provide sufficient conditions for LED policies to be both throughput optimal and delay optimal in heavy-traffic. Our key contributions can be summarized as follows.

First, we introduce the LED framework for designing load balancing policies for heterogeneous systems with multiple dispatchers. In this framework, each dispatcher keeps its own local estimates of queue lengths for all the servers, and makes its dispatching decision based purely on its own local estimates according to a certain dispatching strategy. The local estimates are updated infrequently via an update strategy that is based on communications between dispatchers and servers. 

Second, we derive sufficient conditions for LED policies to be throughput optimal and delay optimal in heavy-traffic. The importance of the sufficient conditions is three-fold: (i) It can be shown that previous local-memory based policies (e.g., LSQ) satisfy our sufficient conditions. As a result, we are able to show that they are not only throughput optimal (in a stronger sense) but also delay optimal in heavy-traffic. (ii) The conditions allow us to design new delay optimal load balancing policies with zero dispatching delay and low message overhead that work for heterogeneous servers and multiple dispatchers. (iii) These conditions also provide us with a systematic approach for generalizing previous optimal policies to the case of multiple dispatchers and exploring the trade-off between memory (i.e., local estimations) and message overhead.

Third, the LED framework also resolves the open problem posed in~\cite{lipshutz2019open}, which asks how to design heavy-traffic delay optimal policies that only use delayed information. Our main results for LED policies not only demonstrate that it is possible to achieve optimal delay in heavy-traffic via only delayed information, but highlight conditions on the extent to which old information is useful. Moreover, they provide methods for using the delayed information to achieve optimality in heavy traffic. Interestingly, the LED framework also shows that, in the case of multiple dispatchers, inaccurate information can actually lead to improved performance.

{To establish the main results, we need to address the following two key challenges. First, each dispatcher in our model has access to delayed and outdated system information. Moreover, each dispatcher does not know the arrivals to servers from other dispatchers, since there is no communication between them. As a result, for throughput optimality, we have to
 carefully design our Lyapunov function, since the local estimates can also be unbounded. For delay optimality, we consider two queueing systems: a local-estimation system and the actual system. Then, we have to transfer the negative drift on the local-estimation system to the actual system, which requires establishing new bounds and the analysis of sample paths.}

\textbf{Related work.}
 The study of efficient load balancing algorithms has been a hot topic for a long time and spans across different asymptotic regimes. The most extensively investigated policy might be Join-Shortest-Queue (JSQ), under which the incoming jobs are always sent to the server with the shortest queue length. JSQ has been shown to be optimal in a stochastic order sense in the non-asymptotic regime for arbitrary arrival and non-decreasing failure rate identical service processes~\cite{weber1978optimal, winston1977optimality}. In the heavy-traffic asymptotic regime, in which the normalized load approaches one and the number of servers is fixed, JSQ has been proved to achieve optimal delay even for heterogeneous servers using both diffusion approximations in~\cite{foschini1978basic} and the recently proposed drift-based method~\cite{eryilmaz2012asymptotically}. However, the optimality of JSQ comes at the cost of a large amount of communication between dispatchers and servers, which is particularly undesirable for large-scale data centers. Thus, some popular low-message overhead alternative policies have been proposed, e.g., Power-of-$d$ and Join-Idle-Queue (JIQ). Under Power-of-$d$, the dispatcher only needs to sample $d \ge 2$ servers and sends arrivals to the server with the shortest queue length among the $d$ samples. This simple policy has been shown to enjoy a doubly exponential decay rate in response time in the large-system asymptotic regime~\cite{mitzenmacher2001power} and achieve optimal delay in heavy-traffic for homogeneous servers~\cite{chen2012asymptotic, maguluri2014heavy}. Another low-message overhead policy is JIQ (or Pull-based policy)~\cite{lu2011join,stolyar2015pull}, under which arrivals are sent to one of the idle servers, if there are any, and to a randomly selected server otherwise. Compared to JSQ and Power-of-$d$, JIQ has the nice property of zero dispatching delay since each arrival can be instantaneously routed rather than waiting for feedback from servers. Moreover, JIQ has been shown to outperform Power-of-$d$ with even smaller message overhead (at most one per job). In particular, under JIQ, arriving jobs achieve asymptotic zero waiting time in the large-system regime while Power-of-$d$ does not. An even stronger result suggests that, in the Halfin-Whitt asymptotic regime, JIQ achieves the same delay performance as JSQ~\cite{mukherjee2016load}. Nevertheless, the performance of JIQ drops substantially in heavy traffic with a finite number of servers, even for homogeneous servers. In fact, it is not heavy-traffic delay optimal in this case~\cite{zhou2017designing}. Motivated by this, recent works have proposed alternative pull-based policies that not only enjoy all the nice features of JIQ but also achieve optimal delay in heavy-traffic~\cite{zhou2017designing,zhou2018heavy}. However, these studies only consider the case of single dispatcher. 

Compared to the large literature on the single dispatcher case, there are only a few works for the scenario of multiple dispatchers, and they mainly focus on the JIQ policy. In particular, \cite{mitzenmacher2016analyzing} presents a new large-system asymptotic analysis of JIQ without the simplifying assumptions in~\cite{lu2011join}. The property of asymptotically zero waiting time of JIQ was generalized to the case of multiple dispatchers in~\cite{stolyar2017pull}. However, the results for JIQ in~\cite{lu2011join, mitzenmacher2016analyzing, stolyar2017pull} all assume that the loads at various dispatchers are strictly equal. Without this assumption, \cite{van2017load} shows that the waiting time under JIQ no longer vanishes in the large-system regime and two enhanced JIQ schemes are proposed. As mentioned earlier, although JIQ is a scalable choice for the multiple-dispatcher case, it is not delay optimal in heavy traffic for homogeneous servers and not even generally stable for heterogeneous systems.

The case of heterogeneous systems with multiple dispatchers has received very little attention from the theoretical community so far. To the best of our knowledge, the framework proposed in~\cite{anonymous2018load} is the first attempt to study efficient load balancing schemes with a theoretical guarantee for the scenario of heterogeneous systems with multiple dispatchers. In particular, under the proposed Loosely-Shortest-Queue (LSQ) framework, each dispatcher independently keeps its own local view of sever queue lengths and routes jobs to the shortest among them. Communication is used only to update the local views and make sure that they are not too far from the real queue lengths. The main contributions of~\cite{anonymous2018load} are the sufficient conditions for any LSQ policy to achieve strong stability with low message overhead. Additionally, extensive simulations have been used to demonstrate its appeal. Nevertheless, a theoretical guarantee on the delay performance of LSQ policies remains an important unsolved question. 

It is worth pointing out that the idea of using local memory to hold possibly old information for load balancing was also explored in two recent works~\cite{jonatha2018power,van2019hyper}. As we discuss later, these two proposed policies are in our LED framework. Both works only consider a single dispatcher and homogeneous servers, which is also a special case of our model. Further, their analysis focuses on the large-system asymptotic regime where the number of servers goes to infinity, while our analysis deals with a finite number of servers.

\section{System Model and Preliminaries}
This section describes the system model and assumptions considered in this paper. Then, several necessary preliminaries are presented.
\subsection{System model}
We consider a discrete-time (i.e., time-slotted) load balancing system consisting of $M$ dispatchers and $N$ possibly-heterogeneous servers. Each server maintains an infinite capacity FIFO queue. At each dispatcher, there is a local memory, through which the dispatcher can have some (possibly delayed) information about the system states. In each time-slot, the central dispatcher routes the new incoming tasks to one of the servers, immediately upon arrival. Once a task joins a queue, it remains in that queue until its service is completed. Each server is assumed to be work conserving, i.e., a server is idle if and only if its corresponding queue is empty.

\subsubsection{Arrivals} Let $A^m(t)$ denote the number of exogenous tasks that arrive at dispatcher $m$ at the beginning of time-slot $t$. We assume that $A_{\Sigma}(t) = \sum_{m=1}^M A^m(t)$ is an integer-valued random variable, which is \emph{i.i.d.} across time-slots. The mean and variance of $A_{\Sigma}(t)$ are denoted by $\lambda_{\Sigma}$ and $\sigma_{\Sigma}^2$, respectively. We further assume that there is a positive probability that $A_{\Sigma}(t)$ is zero. The allocation of total arriving tasks among the $M$ dispatchers is allowed to use any arbitrary policy that is independent of system states. Note that, in contrast to previous works on multiple dispatchers~\cite{lu2011join,mitzenmacher2016analyzing,stolyar2017pull}, we do not require that the loads at all dispatchers are equal. We assume that there is a strictly positive probability for tasks to arrive at each dispatcher at any time-slot $t$. That is, there exists a strictly positive constant $p_0$ such that 
\begin{align}
 	\mathbb{P}\left( A^m(t) > 0\right) \ge p_0, \quad \forall (m,t)\in \mathcal{M}\times\mathbb{N},
\end{align} 
where $\mathcal{M}=\{1,2,\ldots,M\}$. Moreover, we assume that $A^m(t)$ is $\emph{i.i.d}$ across time-slots with mean arrival rate denoted by $\lambda_m$.
We further let $A^m_n(t)$ denote the number of new arrivals at server $n$ from dispatcher $m$ at the beginning of time-slot $t$. Let $A_n(t) = \sum_{m=1}^MA_n^m(t)$ be the total number of arriving tasks at server $n$ at the beginning of time-slot $t$.

\subsubsection{Service} Let $S_n(t)$ denote the amount of service that server $n$ offers for queue $n$ in time-slot $t$. That is, $S_n(t)$ is the maximum number of tasks that can be completed by server $n$ at time-slot $t$. We assume that $S_n(t)$ is an integer-valued random variable, which is \emph{i.i.d.} across time-slots. We also assume that $S_n(t)$ is independent across different servers as well as the arrival process. The mean and variance of $S_n(t)$ are denoted as $\mu_n$ and $\nu_n^2$, respectively. Let $\mu_{\Sigma} \triangleq \Sigma_{n=1}^N \mu_n$ and $\nu_{\Sigma}^2 \triangleq \Sigma_{n=1}^N \nu_n^2$ denote the mean and variance of the hypothetical total service process $S_{\Sigma}(t) \triangleq \sum_{n=1}^N S_n(t)$. Let $\epsilon = \mu_{\Sigma} - \lambda_{\Sigma}$ characterize the distance between the arrival rate and the boundary of capacity region.


\subsubsection{Queue Dynamics} Let $Q_n(t)$ be the queue length of server $n$ at the beginning of time slot $t$. 
Let $A_n(t)$ denote the number of tasks routed to queue $n$ at the beginning of time-slot $t$ according to the dispatching decision. 
Then the evolution of the length of queue $n$ is given by 
\begin{equation}
	\label{eq:Qdynamic}
	Q_n(t+1) = Q_n(t) + A_n(t) - S_n(t) + U_n(t), n = 1,2,\ldots, N,
\end{equation}
where $U_n(t) = \max\{S_n(t)-Q_n(t)-A_n(t),0\}$ is the unused service due to an empty queue.

We do not assume any specific distribution for arrival and service processes. Moreover, in contrast to previous works~\cite{zhou2017designing,eryilmaz2012asymptotically}, we do not require that both arrival and service processes have a finite support. Instead, we only need the condition that their distributions are light-tailed. More specifically, we assume that 
\begin{align}
\label{ass:light-tail}
		\ex{e^{\theta_1 A_{\Sigma}(t)} } \le D_1 \text{ and } \ex{e^{\theta_2 S_{n}(t)} } \le D_2,
\end{align}
for each $n$ where the constants $\theta_1 > 0$, $\theta_2 > 0$, $D_1 <\infty$ and $D_2 <\infty$ are all independent of $\epsilon$.

\subsection{Local-Estimation-Driven (LED) framework}
We are interested in the case that the local memory at each dispatcher $m$ stores an estimate of the queue length for each server $n$. In particular, we let $\widetilde{Q}^m_n(t)$ be the local estimate of the queue length for server $n$ from dispatcher $m$ at the beginning of time-slot $t$ (before any arrivals and departures). More specifically, we introduce the following framework for load balancing.

\begin{definition}
	A Local-Estimation-Driven (LED) policy is composed of the following components:
	\begin{enumerate}[(a)]
		\item \textbf{Dispatching strategy:} At the beginning of each time-slot, each dispatcher $m$ chooses one of the servers for new arrivals \textit{purely} based on its local estimates (i.e., local queue length estimates $\widetilde{\Q}^m$) 
		\item \textbf{Update strategy:} At the end of each time-slot, each dispatcher would possibly update its local estimates, e.g., synchronize local queue length estimate with the true queue length.
	\end{enumerate}
\end{definition}

The definition of LED is broad, and it includes a variety of classical load balancing policies. For example, it can be seen to include LSQ policy studied in~\cite{anonymous2018load}, by choosing the dispatching strategy to be that new arrivals at each dispatcher are dispatched to the queue with the shortest local estimate. Moreover, it also includes two recent local memory based policies in~\cite{jonatha2018power,van2019hyper} that are developed for the case of single dispatcher and homogeneous servers.

To study LED, we model the system as a discrete-time Markov chain $\{Z(t) = (\Q(t),m(t)), t\ge 0\}$ with state space $\mathcal{Z}$, using the queue length vector $\Q(t)$ together with the memory state $m(t) \triangleq (\widetilde{\Q}^1(t), \widetilde{\Q}^2(t),\ldots, \widetilde{\Q}^m(t))$. We consider a set of load balancing systems $\{Z^{(\epsilon)}(t), t\ge 0\}$ parameterized by $\epsilon$ such that the mean arrival rate of the total exogenous arrival process $\{A_{\Sigma}^{(\epsilon)}(t), t\ge 0\}$ is $\lambda_{\Sigma}^{(\epsilon)} = \mu_\Sigma - \epsilon$. Note that the parameter $\epsilon$ characterizes the distance between the arrival rate and the boundary of the capacity region. We are interested in the throughput performance and the steady-state delay performance in the heavy-traffic regime under any LED policy.

A load balancing system is stable if the Markov chain $\{Z(t), t\ge 0\}$ is positive recurrent, and $\overline{Z} = \{\overline{\Q}, \overline{m}\}$ denotes the random vector whose distribution is the same as the steady-state distribution of $\{Z(t), t\ge 0\}$. We have the following definition.
{
\begin{definition}[Throughput Optimality]
	A load balancing policy is said to be throughput optimal if for any arrival rate within the capacity region, i.e., for any $\epsilon > 0$, the system is positive recurrent and all the moments of $\big\lVert{\overline{\Q}^{(\epsilon)}}\big\rVert$ are finite.
\end{definition}

Note that this is a stronger definition of throughput optimality than that in~\cite{anonymous2018load,wang2016maptask,xie2016scheduling} because, besides the positive recurrence, it also requires all the moments to be finite in steady state for any arrival rate within the capacity region.}

To characterize the steady-state average delay performance in the heavy-traffic regime when $\epsilon$ approaches zero, by Little's law, it is sufficient to focus on the summation of all the queue lengths. First, recall the following fundamental lower bound on the expected sum queue lengths in a load balancing system under any throughput optimal policy \cite{eryilmaz2012asymptotically}. Note that this result was originally proved with the assumption of finite support on the service process (Lemma 5 in~\cite{eryilmaz2012asymptotically}), which can be generalized to service processes with light-tailed distributions with a careful analysis of the unused service, see our proof of Lemma~\ref{lem:unused_service}.

\begin{lemma}
\label{lem:lower_bound}
    Given any throughput optimal policy and assuming that $(\sigma_{\Sigma}^{(\epsilon)})^2$ converges to a constant $\sigma_{\Sigma}^2$ as $\epsilon$ decreases to zero, then 
	\begin{equation}
	\label{eq:lower_bound}
		\liminf_{\epsilon \downarrow 0} \epsilon \ex{\sum_{n=1}^N \overline{Q}_n^{(\epsilon)} } \ge \frac{\zeta}{2},
	\end{equation}
	where $\zeta \triangleq \sigma_{\Sigma}^2 + \nu_{\Sigma}^2$.
\end{lemma}

The right-hand-side of Eq. \eqref{eq:lower_bound} is the heavy-traffic limit of a hypothesized single-server system with arrival process $A_\Sigma^{(\epsilon)}(t)$ and service process $\sum_n^N S_n(t)$ for all $t\ge0$. This hypothetical single-server queueing system is often called the \textit{resource-pooled system}. Since a task cannot be moved from one queue to another in the load balancing system, it is easy to see that the expected sum queue lengths of the load balancing system is larger than the expected queue length in the resource-pooled system. However, if a policy achieves the lower bound in Eq. \eqref{eq:lower_bound} in the heavy-traffic limit, based on Little's law this policy achieves the minimum average delay of the system in steady-state, and thus said to be heavy-traffic delay optimal, see~\cite{eryilmaz2012asymptotically,maguluri2014heavy,wang2016maptask,xie2015priority,xie2016scheduling,zhou2017designing}.

\begin{definition}[Heavy-traffic Delay Optimality in Steady-state]
	A load balancing scheme is said to be heavy-traffic delay optimal in steady-state if the steady-state queue length vector $\overline{\Q}^{(\epsilon)}$ satisfies 
	\begin{equation*}
		\limsup_{\epsilon \downarrow 0} \epsilon \ex{\sum_{n=1}^N \overline{Q}_n^{(\epsilon)} } \le \frac{\zeta}{2},
	\end{equation*}
	where $\zeta$ is defined in Lemma \ref{lem:lower_bound}.
\end{definition}

\subsection{Dispatching Preference}
In order to provide a unified way to specify the dispatching strategy in LED, we first introduce a concept called \emph{dispatching preference}. In particular, let $P^m_n(t)$ be the probability that new arrivals at dispatcher $m$ are dispatched to server $n$ at time-slot $t$. We define $\beta_n^m(t) \triangleq P^m_n(t) - \frac{\mu_n}{\mu_{\Sigma}}$, which is the difference in probability that server $n$ will be chosen under a particular dispatching strategy and random routing (weighted by service rate). Then, we have the following definition.

\begin{definition}[Dispatching preference]
	Fix a dispatcher $m$, let $\sigma_t(\cdot)$ be a permutation of $(1,2,\ldots,N)$ that satisfies 
\begin{align*}
	\widetilde{Q}^m_{\sigma_t(1)}(t) \le \widetilde{Q}^m_{\sigma_t(2)}(t) \le \ldots \le \widetilde{Q}^m_{\sigma_t(N)}(t).
\end{align*}
The dispatching preference at dispatcher $m$ is a $N$-dimensional vector denoted by $\Delta^m(t)$, the $n$th component of which is given by $\Delta^m_n(t) \triangleq \beta^m_{\sigma_t(n)}$.
\end{definition}
In words, the dispatching preference at a dispatcher $m$ specifies how servers with different local estimates are preferred in a unified way such that it is independent of the actual values of local estimates. It only depends on the relative order of local estimates. More specifically, fix a dispatcher $m$, by definition we can see that weighted random routing strategy has no preference for any servers and $\Delta_n^m(t) = 0$ for any $n$. On the other hand, if new arrivals are always dispatched to the server with the shortest local estimate (e.g, LSQ policy), we have $\Delta^m_1(t) > 0$ and $\Delta^m_n(t) < 0$ for all $2\le n\le N$. Thus, we can see that a positive value for $\Delta_n^m(t)$ means that the dispatching strategy has a preference for the server with the $n$th shortest local estimation. This observation directly motivates the following two definitions.

\begin{definition}[Tilted dispatching strategy]
	A dispatching strategy adopted at dispatcher $m$ is said to be \textit{tilted} if there exists a $k \in \{2,3,\ldots N\}$ such that for all $t$, $\Delta_n^m(t) \ge 0$ for all $n\le k$ and $\Delta_n^m(t) \le 0$ for all $n \ge k$.
\end{definition}

\begin{definition}[$\delta$-tilted dispatching strategy]
	A dispatching strategy adopted at dispatcher $m$ is said to be \textit{$\delta$-tilted} if for all $t$ (i) it is a tilted dispatching strategy and (ii) there exists a positive constant $\delta$ such that $\Delta_1^m(t) \ge \delta$ and $\Delta_N^m(t) \le -\delta$.
\end{definition}

{
\begin{remark}
	Note that similar definitions were first provided in~\cite{zhou2017designing} for the case of a single dispatcher with up-to-date information. Based on these definitions, sufficient conditions were presented for throughput and heavy-traffic optimality. However, these conditions cannot be directly applied to our model due to the following two major challenges. One is that, in our model, each dispatcher only has access to outdated information. The other is that each dispatcher has no idea of the arrivals at the servers coming from other dispatchers, since there is no communication between them. To handle these challenges, we have to develop new techniques.
\end{remark}
}

We end this section by providing intuitions behind the two definitions. To start, it can be seen easily that $\sum_{n=1}^N \Delta_n^m(t) = 0$ for all $m$ and $t$ via the definition of dispatching preference. Roughly speaking, a tilted dispatching strategy means that compared to (weighted) random routing (which does not have any preference), the probabilities of choosing servers with shorter local estimates (the first $k$ shortest ones) are increased, and, as a result, the probabilities of choosing servers with longer local estimates are reduced. This is the reason why we call it tilted, since more preference is given to queues with shorter local estimates. Therefore, a tilted dispatching strategy can be viewed as a strategy that is as least as `good' as (weighted) random routing. On the other hand, a $\delta$-tilted dispatching strategy can be viewed as a strategy that is strictly better than (weighted) random routing. The reason is that, besides the fact that it is tilted, it also requires that there is a strictly positive preference of the server with the shortest local estimation. 

\section{Main Results}
In this section, we first present the sufficient conditions for LED policies to be throughput optimal and heavy-traffic delay optimal. Then, we explore several example policies within LED framework to demonstrate its flexibility in designing new load balancing schemes. 

\subsection{Sufficient Conditions}
Let us begin with the sufficient conditions for LED policies to be throughput optimal. In particular, we specify conditions for the dispatching strategy and update strategy that guarantee throughput optimality.

To state the theorem, we need the following notation. Let $\mathcal{I}^m_n(t)$ be an indicator function which equals $1$ if and only if the local estimate of server $n$'s queue length at dispatcher $m$ gets updated, i.e., the estimated queue length $\widetilde{Q}^m_n(t)$ is set to the actual queue length $Q_n(t)$ at the end of time-slot $t$. 

\begin{theorem}
\label{thm:throughput}
	Consider an LED policy. Suppose the dispatching strategy at each dispatcher is tilted and the update strategy can guarantee the condition that there exists a positive constant $p$ such that
	\begin{align}
	\label{eq:throughput_condition}
		\ex{\mathcal{I}^m_n(t)\mid Z(t) = Z} \ge p 
	\end{align}
	holds for all $Z$ and $(m,n,t) \in \mathcal{M}\times\mathcal{N}\times \mathbb{N}$. 
	Then, this policy is throughput optimal, i.e., the system under this policy is positive recurrent with all the moments being bounded for any $\epsilon>0$.
\end{theorem}
\begin{proof}
	See Section~\ref{sec:proof_throughput}
\end{proof}
Note that this theorem directly implies that LSQ is not only strongly stable but also enables the system to have all the moments bounded in steady-state. Moreover, it suggests that any dispatching strategy that is as good as (weighted) random routing is sufficient to guarantee throughput optimality. Further, the update probability can be a function of the traffic load. 

Now, we turn to presenting the sufficient conditions for LED policies to be delay optimal in heavy traffic. In order to achieve delay optimality, we need stronger conditions on both the dispatching strategy and the update strategy. 

\begin{theorem}
\label{thm:delay}
	Consider an LED policy. Suppose the dispatching strategy at each dispatcher is $\delta$-tilted with a uniform lower bound $\delta >0$ being independent of $\epsilon$. Suppose the update strategy can guarantee that there exists a positive constant $p$ (independent of $\epsilon$) such that
	\begin{align}
	\label{eq:delay}
		\ex{\mathcal{I}^m_n(t)\mid Z(t) = Z} \ge p 
	\end{align}
	holds for all $Z$ and $(m,n,t) \in \mathcal{M}\times\mathcal{N}\times \mathbb{N}$. 
	Then, this policy is heavy-traffic delay optimal.
\end{theorem}
\begin{proof}
	See Section~\ref{sec:proof_delay}
\end{proof}
This theorem not only establishes a delay performance guarantee for many previous local-memory based policies (e.g., LSQ in~\cite{anonymous2018load}, low-message policies in~\cite{jonatha2018power,van2019hyper}), but provides us with the flexibility to design new delay optimal load balancing for different scenarios with heterogeneous servers and multiple dispatchers, as discussed in the next section. More importantly, our results directly suggest that it is possible to use only delayed information to achieve delay optimality, which resolves one of the open problems listed in~\cite{lipshutz2019open}.

\textbf{High-level proof idea.} We end this section by providing drift-based intuitions behind the technical proofs. In particular, let us consider two queueing systems: a local-estimation system and the actual system (i.e., queue lengths at servers). For throughput optimality, it requires the actual system to have a drift towards the origin. First, by the definition of a tilted dispatching strategy, it provides an equivalent drift on the local-estimation system that is towards the origin. Then, the condition on the update strategy guarantees that the local-estimation system is not too far away from the actual system. Hence, the actual system also has a drift towards the origin. The heavy-traffic delay optimality not only requires a drift towards the origin, but also needs a drift towards the line that all the queue lengths are equal. First, by the definition of a $\delta$-tilted dispatching strategy, there is a drift towards the line that all the local estimates are equal within a given dispatcher. Then, by the condition for the update strategy, the drift on the local-estimation system can be transfered to a drift on the actual system, and hence delay optimality. Note that, in the current proof, in order to make this `drift-transfer' process valid, we impose the condition that both $\delta$ and $p$ are independent of $\epsilon$, which is not necessarily required and both of them could possibly be a particular function of $\epsilon$ as in~\cite{zhou2018flexible}. This relaxation could be an interesting future research direction.

\subsection{Examples}
To illustrate the applications of Theorems~\ref{thm:throughput} and~\ref{thm:delay}, in this section, we introduce examples of LED policies that are both throughput optimal and heavy-traffic delay optimal. The flexibility provided by our sufficient conditions not only allows us to include previous policies as special cases, but enables us to design new flexible policies.

Let us first introduce some typical $\delta$-tilted dispatching strategies.
\begin{example}[Local--Join-Shortest-Queue (L-JSQ)]
	At the beginning of each time-slot $t$, the dispatcher forwards its arrivals to the server with the shortest local estimation with ties broken arbitrarily. That is, consider dispatcher $m$, the chosen server is $i^* \in \argmin_n\{\widetilde{Q}_n^m\}$ .
\end{example}

This dispatching strategy is the same as that in the LSQ policy in~\cite{anonymous2018load}. By the definition of dispatching preference, we can see that under L-JSQ, $\Delta_1^m(t) = 1-{\mu_{\sigma_t(1)}}/{\mu_{\Sigma}} > 0$ and $\Delta_n^m(t) = -\mu_{\sigma_t(n)}/{\mu_{\Sigma}} < 0$. Hence, it is $\delta$-tilted even for heterogeneous servers with $\delta = \mu_{min}/\mu_{\Sigma}$ where $\mu_{min} = \min_n\mu_n$.

Instead of always joining the server with the shortest local estimate, it is also possible to join a sever whose queue length is below a threshold while satisfying the condition of $\delta$-tilted dispatching preference.
\begin{example}[Local--Join-Below-Average (L-JBA)]
	At the beginning of each time-slot $t$, the dispatcher forwards its arrivals to a randomly chosen server whose local estimate is below or equal to the average local queue length estimation. That is, consider dispatcher $m$ with the average local estimate being $\bar{Q}^m(t) = \frac{1}{N}\sum_n\widetilde{Q}_n^m(t)$. Let $\mathcal{A} \triangleq \{n: \widetilde{Q}_n^m(t) \le \bar{Q}^m(t)\}$. Then, for each $i\in \mathcal{A}$, $P_i^m(t) = \mu_i/\sum_{n\in\mathcal{A}}\mu_n$, and for $i \notin \mathcal{A}$, $P_i^m(t) = 0$.
\end{example}

It can be easily shown from the definition that L-JBA is also $\delta$-tilted. Note that, compared to L-JSQ, in the heterogeneous case, it needs the dispatcher to know the service rate of each server, which can be easily obtained by the update strategies introduced next. This strategy is more flexible than L-JSQ since it does not require new arrivals to be only sent to the server with the shortest local estimate, which could be used in the scenarios with data locality. Moreover, some randomness in the dispatching strategy is also useful, as discussed in the next section. 

Further, it is possible to generalize many previous heavy-traffic delay optimal policies into the LED framework. For example, we can directly apply the Power-of-$d$ policy as our dispatching strategy.
\begin{example}[Local--Power-of-$d$ (L-Pod)]
	At the beginning of each time-slot $t$, the dispatcher randomly chooses $d \ge 2$ servers and sends arrivals to the server that has the shortest local estimation among the $d$ servers. 
\end{example}

It can be easily shown that L-Pod is tilted for homogeneous servers. Moreover, for a given $m$, we have $\Delta_1^m(t) = \frac{d-1}{N}$ and $\Delta_N^m(t) = -\frac{1}{N}$, and hence it is $\delta$-tilted with $\delta = \frac{1}{N}$.

Now, let us turn to discussing update strategies that satisfy the condition in Theorem~\ref{thm:delay}. In particular, the update strategy can either be push-based (dispatcher samples servers) or pull-based (servers report to dispatchers). 

\begin{definition}[Push-Update]
	If there are new arrivals, then at the end of the time-slot the dispatcher $m$ samples $d$ distinct servers with a positive probability $\hat{p}$. Then, it updates the corresponding $d$ local estimations with the true values.
\end{definition}

It has been shown in~\cite{anonymous2018load} that even for $d = 1$, the push-update strategy is guaranteed to satisfy the condition in Theorem~\ref{thm:delay}.

\begin{definition}[Pull-Update]
	At the end of each time-slot, for each server $n$ if there are completed tasks, then the server will uniformly at random pick a dispatcher $m$ and then abide by one of the following two rules:
	\begin{itemize}
		\item If the server becomes idle (i.e., no tasks), it sends $(n, 0)$ to dispatcher $m$.
		\item If not, it sends $(n, Q_n)$ to dispatcher $m$ with probability $\hat{p}$.
	\end{itemize}
\end{definition}

It has been shown in~\cite{anonymous2018load} that for any $\hat{p}>0$, the pull-update strategy is guaranteed to satisfy the condition in Theorem~\ref{thm:delay}.

Now, having introduced both the dispatching strategy and the update strategy, we can combine them to obtain different LED policies that are delay optimal in heavy-traffic. For example, we have L-JSQ-Push, L-JSQ-Pull, L-JBA-Push, L-JBA-Pull for heterogeneous servers, as well as L-Pod-Push and L-Pod-Pull for homogeneous servers. 

We end this section by summarizing the contributions of the LED framework. (i) \textbf{It covers previous polices.} L-JSQ-Push (with $\hat{p} = 1$) and L-JSQ-Pull  are the same as LSQ policies considered in~\cite{anonymous2018load}, which include the policies developed in both~\cite{jonatha2018power} and~\cite{van2019hyper} as special cases. Thus, by Theorems~\ref{thm:throughput} and~\ref{thm:delay}, all these policies are throughput and heavy-traffic delay optimal. (ii) \textbf{It allows randomness in dispatching.} The randomness introduced in L-JBA and L-Pod is helpful when dealing with the scenario with an extreme low budget on the message overhead, as discussed next. (iii) \textbf{It enables trade-offs between memory and message overhead.} For example, L-Pod-Push and L-Pod-Pull represent good examples that trade memory for low message overhead. That is, if each dispatcher directly uses the traditional Power-of-$d$ without any memory, then at least $4$ messages needed to guarantee delay optimality in heavy-traffic. In contrast, in both L-Pod-Push and L-Pod-Pull, the \emph{worst-case} message overhead is just $1$ per arrival. 
In addition, the message can be further reduced by choosing a smaller value of $\hat{p}$ in the update strategy.

\section{Discussion}
Before moving to the proofs, we would like to discuss key features and insights about LED, and point out possible refinements on LED.

\vspace{-1.3mm}
\subsection{Key features of LED}
In this section, we highlight the key features of the LED framework, including low message overhead, zero dispatching delay, low computational complexity and appealing performance across various loads.

\textbf{Low message overhead.} It should be noted that the communication overhead occurs only during the update phase in LED policies. For the push-update strategy, the number of messages per arrival is at most $2d$ ($d$ can even be one). For the pull-update strategy, the number of messages per arrival is at most $1$. In contrast, JSQ needs $2N$ messages per arrival and Power-of-$d$ needs at least $4$ messages per arrival. Although JIQ has a comparative worst-case message overhead as LED policies, it is not stable for heterogeneous servers. 

\textbf{Zero dispatching delay.} Another key feature of all LED policies is that there is zero dispatching delay. That is, the dispatcher can immediately route its new arrivals to the chosen server since the decision is made purely based on its local estimations. Moreover, the communication between dispatchers and servers happens only after the decision is made. This is in contrast to typical push-based policies like JSQ and Power-of-$d$, under which the dispatcher has to wait for the response of sampled servers to make its dispatching decision, resulting in a non-zero dispatching delay.

\textbf{Low computational complexity.} In order to implement LED policies, each dispatcher has to keep an array of size $N$ its local estimations. Such a space requirement is negligible in a modern cluster. Further, the operations required by dispatching strategies of LED policies are very efficient. For example, in order to find the server with the minimal local estimate in L-JSQ, we can keep the array in a min-heap data structure. For L-JBA, we can calculate the average by using an efficient running average algorithm. For the simple L-Pod, it only needs random number generators.

\textbf{Appealing performance across loads.} Although the theoretical delay optimality for the LED framework holds in the heavy-traffic asymptotic regime, the family of LED policies includes efficient policies that significantly outperform alternative low-message overhead policies with the same (or even smaller) amount of communications. For example, if the dispatching strategy adopts L-JSQ in LED, then it reduces to the LSQ policy proposed in~\cite{anonymous2018load}, which appeals to enjoy good performance over a wide range of traffic loads in different scenarios via extensive simulations.

As mentioned earlier, the class of heavy-traffic delay optimal LED policies is broad and includes flexible choices of different dispatching and update strategies based on different application scenarios. The actual delay performance (except the heavy-load scenario) varies with the particular choice of dispatching strategy or update strategy under different scenarios. Thus, it is not possible to pick one particular LED policy that fits every circumstance, which is also not the focus of this paper. Instead, it would be useful to present some useful insights about the LED framework, as presented in the following. These insights could serve as the guidance on the choice or design of new LED policies.

\subsection{Useful insights from LED}
The main trait of the LED framework is that only local, possibly delayed and inaccurate information, is used for making the dispatching decision. In the following, we present two useful insights about the use of inaccurate delayed information for load balancing.

\begin{figure}[t]
	\graphicspath{{./Figures/}}
	\centering
	\includegraphics[width=3.5in]{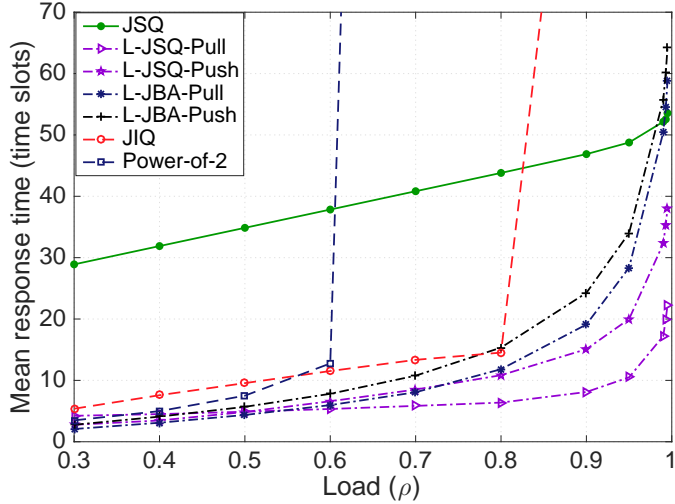}
	\caption{Inaccurate information could improve performance in multiple-dispatcher case.}\label{fig:insight1}
\end{figure}

\textbf{Inaccurate information can improve performance.} A big problem for load balancing with multiple dispatchers is \emph{herd behavior}, which means that arrivals at different dispatchers join the same server. This often leads to a poor delay performance in practice~\cite{suresh2015c3}.  For example, JSQ used in the case of multiple dispatchers leads to a serious herd behavior since all the dispatchers will route arrivals to the single shortest queue. In contrast, under the LED framework, each dispatcher may believe that a different queue is the shortest according to its own local estimates because these estimates are inaccurate and delayed. Thus, jobs at different dispatchers are sent to different queues that may not have the actual shortest length but still have relatively small queue lengths. This intuition is illustrated by Fig.~\ref{fig:insight1}. In particular, we consider a set up with $10$ dispatchers and $100$ heterogeneous servers. All the LED policies are configured to have the same average message overheads as Power-of-$2$. It can be seen that the LED policies are not only stable but achieve a much better performance compared to JSQ, which suffers from the herd behavior in the multiple-dispatcher case.

\textbf{Randomness is useful for heavily-delayed information.} As mentioned earlier, the LED framework provides us with the possibility of exploring load balancing with extremely low message overhead by choosing a small value $\hat{p}$ in the update strategy. As a result, the local information at each dispatcher will only be updated after a long time interval. In this case, if a deterministic dispatching strategy (e.g., L-JSQ) is adopted, it would again incur herd behavior (even for a single dispatcher case) since all the arrivals during the long update interval will join the same queue. This is another motivation for considering  L-JBA and L-Pod, which naturally introduce a certain level of randomness and hence help avoid the herd behavior as suggested by~\cite{mitzenmacher2000useful}. To illustrate this insight, we consider a set up with $10$ dispatchers and $100$ homogeneous servers. We compare the delay performance of L-JSQ-Push, L-Pod-Push and L-JBA-Push with the update probability set to $\hat{p}=0.01$  and $d=2$. As shown in Fig.~\ref{fig:insight2}, both L-JBA-Push and L-Pod-Push outperforms L-JSQ-Push, which suffers from herd behavior because of heavily-delayed information. 
\begin{figure}[t]
	\graphicspath{{./Figures/}}
	\centering
	\includegraphics[width=3.5in]{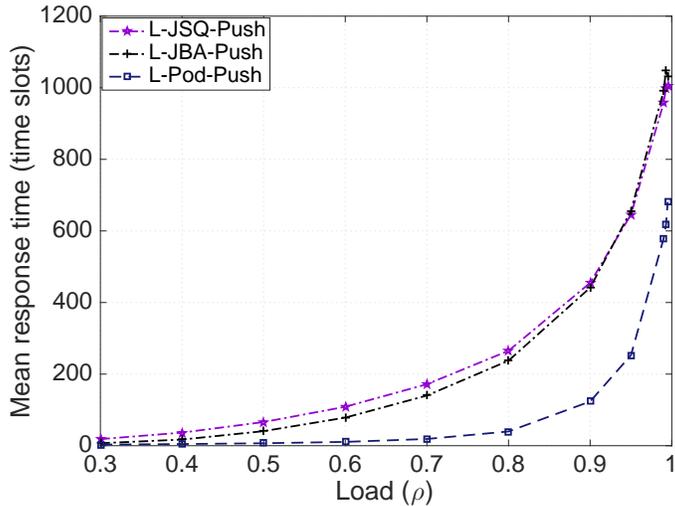}
	\caption{Randomness is useful for heavily-delayed information.}\label{fig:insight2}
\end{figure}

\subsection{Refinements on LED}
Our main results suggest that there is a large class of heavy-traffic delay optimal LED policies. On the one hand, it provides us with flexibility to tailor our policy design for different application scenarios with different choices of dispatching and update strategies. On the other hand, it also suggests the need for refinements on LED beyond delay optimality in heavy-traffic. To this end, we introduce two possible directions for refinements.

\textbf{Degree of queue imbalance.} As introduced in~\cite{zhou2018degree}, \emph{degree of queue imbalance} is a refined metric to further distinguish heavy-traffic delay optimal policies. The idea is that, instead of looking at the average queue length (and hence average delay), the degree of queue imbalance measures the expected difference in queue lengths among the servers. By following the proof of Proposition 5.6 in~\cite{zhou2018degree}, we can establish that the degree of queue imbalance of all heavy-traffic delay optimal LED policies is $O(\frac{1}{\delta^2p^4})$. Thus, even though by Theorem~\ref{thm:delay}, any positive $\delta$ and $p$ are sufficient for delay optimality in heavy-traffic, a dispatching strategy with smaller $\delta$ or a update strategy with a smaller $p$ could affect the performance in practice.

\textbf{Other asymptotic regimes.} In this paper, we focus on the heavy-traffic asymptotic regime where the number of servers is fixed and the load approaches one. As mentioned before, there are also other asymptotic regimes in the analysis of load balancing schemes. One possible direction is to extend the fluid-limit techniques for the large-system regime in~\cite{van2019hyper} to the case of multiple dispatchers and heterogeneous servers. Another alternative regime is the many-server heavy-traffic regime (e.g., Halfin-Whitt regime), which tends to keep a balance between heavy-traffic regime and large-system regime. Studying LED in such a regime is another interesting direction for future work.

\section{Proofs}
In this paper, we extend the Lyapunov drift-based approach developed in~\cite{eryilmaz2012asymptotically} to allow for unbounded supports of arrival and service processes. In particular, we replace the finiteness condition on the drift in~\cite{eryilmaz2012asymptotically} by a stochastically dominated condition, as shown in (C2) in Lemma~\ref{lem:basis}.
As proved in~\cite{hajek1982hitting}, this weaker condition, combined with a negative drift condition, can still guarantee finite moment bounds. Besides a weaker condition, we also replace the one-step drift with a $T$-step drift. Formally, we use the following lemma to derive bounded moments in steady state.
\begin{lemma}
	      \label{lem:basis}
	        For an irreducible aperiodic and positive recurrent Markov chain $\{X(t), t \ge 0\}$ over a countable state space $\mathcal{X}$, which converges in distribution to $\overline{X}$,  and suppose $V: \mathcal{X} \rightarrow \mathbb{R}_{+}$ is a Lyapunov function. We define the $T$ time slot drift of $V$ at $X$ as 
	        \[\Delta V(X)\triangleq [V(X(t_0+T)) - V(X(t_0))] \mathcal{I}(X(t_0) = X),\]
	        where $\mathcal{I}(.)$ is the indicator function. Suppose for some positive finite integer $T$, the $T$ time slot drift of $V$ satisfies the following conditions:

	        \begin{itemize}
	          \item (C1) There exists an $\eta> 0$ and a $\kappa <  \infty$ such that for any $t_0 = 1,2,\ldots$ and for all $X \in \mathcal{X}$ with $V(X)\ge \kappa$, 
	          \[\mathbb{E}\left[\Delta V(X) \mid X(t_0) = X\right]\le -\eta.\]
	          \item (C2) $|\Delta V(X)| \prec W$ for all $t_0$ and all $X \in \mathcal{X}$, and $\ex{e^{\theta W}} = D$ is finite for some $\theta > 0$,
	        \end{itemize}

	        Then $\{V(X(t)), t\ge0\}$ converges in distribution to a random variable $\overline{V}$ for which there exists a $\theta^* >0$ and a $C^* < \infty$ such that 
	        \begin{align*}
	        	\ex{e^{\theta^* \overline{V}}} \le C^*,
	        \end{align*}
	        which directly implies that all the moments of $\overline{V}$ exist and are finite.
\end{lemma}

\subsection{Proof of Theorem~\ref{thm:throughput}}
\label{sec:proof_throughput}
To start with, let us first show that the Markov chain $\{Z(t) = (\Q(t), m(t)), t\ge 0\}$ with $m(t) \triangleq (\widetilde{\Q}^1(t), \widetilde{\Q}^2(t),\ldots, \widetilde{\Q}^m(t))$ is irreducible and aperiodic. Let the initial state be $Z(0) = (\Q(0),m(0)) = (0_{1\times N}, 0_{1\times MN})$ and the state space $\mathcal{Z}$ consists of all the states that can be reached from the initial state. Consider any state $Z$, the queue length vector $\Q$ can reach the initial state with a positive probability since the event that there are no exogenous arrivals and all the offered service is at least one during each time-slot happens with positive probability under our assumptions. Moreover, under the condition for the update strategy given by Eq.~\eqref{eq:throughput_condition}, the event that $\Q$ remains as the initial state while all $\widetilde{\Q}^m$ reach to the initial state happens with a positive probability. Therefore, any state in the state space can reach the initial state, and hence the Markov chain is irreducible. The aperiodicity of the Markov chain comes from the fact that the transition probability from the initial state to itself is positive. 

In order to show positive recurrence, we adopt the Foster-Lyapunov theorem. In particular, we consider the following Lyapunov function $W(Z(t)) = \norm{\Q(t)}^2 + \sum_{m=1}^M\norm{\Q(t) - \widetilde{\Q}^m(t)}_1$, and in the rest of the proof we use $W(t)$ as an abbreviation of $W(Z(t))$ Let $X_n^m(t) \triangleq |Q_n(t) - \widetilde{Q}_n^m(t)|$. The conditional mean drift of $W(t)$ defined as $D(Z(t_0)) \triangleq \ex{W(t_0+T) - W(t_0)\mid Z(t_0)}$ can be decomposed as follows
\begin{align}
\label{eq:drift_all}
	D(Z(t_0)) = D_Q(t_0) + \sum_{m=1}^M\sum_{n=1}^N D_{X^m_n}(t_0) 
\end{align}
where 
\begin{align*}
	D_Q(t_0) &\triangleq \ex{\norm{\Q(t_0+T)}^2 - \norm{\Q(t_0)}^2 \mid Z(t_0)}\\
	D_{X^m_n}(t_0)  &\triangleq \ex{ X_n^m(t_0+T) - X_n^m(t_0)\mid Z(t_0)}
\end{align*}
Let us first consider the tern $D_{X^m_n}(t_0)$. Note that for all $t_0$, $m$ and $n$
\begin{align}
	&\ex{X_n^m(t_0+1) \mid Z(t_0) = Z}\nonumber\\
	\le&\ex{ (1-\mathcal{I}_n^m(t_0))\left(X_n^m(t_0) + A_n(t_0) + S_n(t_0)\right)\mid Z(t_0) = Z}\nonumber\\
	\lep{a}&(1-p)X_n^m(t_0) + \lambda_{\Sigma} + \mu_{max}\label{eq:drift_X}
\end{align}
where (a) follows from the condition in Eq.~\eqref{eq:throughput_condition} and $\mu_{max} = \max_n \mu_n$.
Then, we have (the time reference $t_0$ is dropped for simplicity)
\begin{align}
	&D_{X^m_n}(t_0)\nonumber\\
	=&\ex{\sum_{t=t_0}^{t_0+T-1} X_n^m(t+1) - X_n^m(t) \mid Z(t_0) = Z}\nonumber\\
	=&\sum_{t=t_0}^{t_0+T-1}\ex{ \ex{X_n^m(t+1) - X_n^m(t) \mid Z(t)}\mid Z}\nonumber\\
	\lep{a}&\sum_{t=t_0}^{t_0+T-1}\ex{-pX_n^m(t) + \lambda_{\Sigma} + \mu_{max} \mid Z}\nonumber\\
	\le&-pX_n^m(t_0) + \lambda_{\Sigma} + \mu_{max}\label{eq:D_X},
\end{align}
where (a) follows from Eq.~\eqref{eq:drift_X}.
Let us turn to consider the term $D_Q(t_0)$. By the queue dynamics in Eq.~\eqref{eq:Qdynamic}, 
\begin{align}
\label{eq:driftW}
	&D_Q(t_0)\nonumber\\
	=&\ex{\sum_{t=t_0}^{t_0+T-1} \norm{\Q(t+1)}^2 - \norm{\Q(t)}^2 \mid Z(t_0) = Z}\nonumber\\
	=&\ex{\sum_{t=t_0}^{t_0+T-1} \norm{\Q(t) + \A(t) -\s(t) +\UU(t)}^2 - \norm{\Q(t)}^2 \mid Z}\nonumber\\
	\lep{a}&\ex{\sum_{t=t_0}^{t_0+T-1}\norm{\Q(t) + \A(t) -\s(t)}^2 - \norm{\Q(t)}^2\mid Z}\nonumber\\
	=& \ex{\sum_{t=t_0}^{t_0+T-1}2\inner{\Q(t)}{\A(t)-\s(t)} + \norm{\A(t) -\s(t)}^2\mid Z}\nonumber\\
	\lep{b} & \ex{\sum_{t=t_0}^{t_0+T-1}2\inner{\Q(t)}{\A(t)-\s(t)} + K\mid Z},
\end{align}
where (a) follows from the facts that $Q_n(t) + A_n(t) -S_n(t) + U_n(t) = \max(Q_n(t) + A_n(t) -S_n(t),0)$ for any $t\ge 0$, and $\left(\max(a,0) \right)^2 \le a^2$ for any $a \in \mathbb{R}$; (b) holds by our assumption of light-tailed distributions for the total arrival process and each service process in Eq.~\eqref{ass:light-tail}. In particular, we have that the second moments for total arrival process and service process of each server are finite (independent of $\epsilon$), and hence there exists a finite upper bound $K$ which is independent of the load parameter $\epsilon$.

Now, let us continue to work on Eq.~\eqref{eq:driftW}. In particular, we have 
\begin{align}
	&\ex{\sum_{t=t_0}^{t_0+T-1} \inner{\Q(t)}{\A(t) -\s(t)} \mid Z(t_0) = Z}\nonumber\\
	=&\sum_{t=t_0}^{t_0+T-1} \ex{\ex{\inner{\Q(t)}{\A(t) -\s(t)}\mid Z(t)}\mid Z(t_0) = Z}\nonumber\\
	=& \sum_{t=t_0}^{t_0+T-1}\ex{\ex{\inner{\Q(t)}{\A(t)}\mid Z(t)}\mid Z(t_0) = Z}\label{eq:arrival}\\
	&- \sum_{t=t_0}^{t_0+T-1} \ex{\sum_{n=1}^N Q_n(t)\mu_n \mid Z(t_0) = Z}.\label{eq:service}
\end{align}
For Eq.~\eqref{eq:arrival}, we have 
\begin{align*}
& \sum_{t=t_0}^{t_0+T-1}\ex{\ex{\inner{\Q(t)}{\A(t)}\mid Z(t)}\mid Z(t_0) = Z}\nonumber\\
=&\sum_{t=t_0}^{t_0+T-1}\ex{ \sum_{n=1}^N Q_n(t) \sum_{m=1}^M \ex{A_n^m(t) \mid Z(t)}\mid Z}\nonumber\\
=&\sum_{t=t_0}^{t_0+T-1}\ex{\sum_{n=1}^N Q_n(t)\sum_{m=1}^MP_n^m(t)\lambda_m \mid Z(t_0) = Z}\nonumber\\
\ep{a}&\sum_{t=t_0}^{t_0+T-1}\ex{\sum_{n=1}^N Q_{n}(t)\sum_{m=1}^M\left(\beta_n^m(t) + \frac{\mu_{n} }{\mu_{\Sigma}} \right)\lambda_m \mid Z(t_0) = Z},\nonumber\\
\end{align*}
where (a) follows from the definition of $\beta^m_n(t)$. Then, it can be further simplified as follows.
\begin{align}
& \sum_{t=t_0}^{t_0+T-1}\ex{\ex{\inner{\Q(t)}{\A(t)}\mid Z(t)}\mid Z(t_0) = Z}\nonumber\\
\ep{a}&\sum_{t=t_0}^{t_0+T-1}\ex{\sum_{n=1}^N Q_{n}(t)\sum_{m=1}^M\beta_n^m(t)\lambda_m\mid Z}\nonumber\\
&+ \sum_{t=t_0}^{t_0+T-1}\ex{\sum_{i=1}^N Q_{n}(t)\sum_{m=1}^M\frac{\mu_{n} }{\mu_{\Sigma}}\left(\mu_{\Sigma} - \epsilon\right)p_m\mid Z}\nonumber\\
=&\sum_{t=t_0}^{t_0+T-1}\ex{\sum_{n=1}^N Q_{n}(t)\sum_{m=1}^M\beta_n^m(t)\lambda_m\mid Z}\nonumber\\
 & +\sum_{t=t_0}^{t_0+T-1} \ex{\sum_{n=1}^N Q_n(t)\mu_n \mid Z} - \sum_{t=t_0}^{t_0+T-1} \ex{\sum_{n=1}^NQ_n(t)\frac{\epsilon \mu_n}{\mu_{\Sigma}}\mid Z},\label{eq:sim_arrival}
\end{align}
where in (a), $p_m$ is the probability that arrivals are allocated to dispatcher $m$ (or it can be viewed as the fraction of the total arrivals that are allocated to dispatcher $m$).

Combining Eqs.~\eqref{eq:arrival},~\eqref{eq:service} and~\eqref{eq:sim_arrival}, yields
\begin{align}
	&\ex{\sum_{t=t_0}^{t_0+T-1} \inner{\Q(t)}{\A(t) -\s(t)} \mid Z(t_0) = Z}\nonumber\\
	=&\sum_{t=t_0}^{t_0+T-1}\ex{\sum_{n=1}^N Q_{n}(t)\sum_{m=1}^M\beta_n^m(t)\lambda_m\mid Z(t_0) = Z}\nonumber \\
	&- \sum_{t=t_0}^{t_0+T-1} \ex{\sum_{n=1}^NQ_n(t)\frac{\epsilon \mu_n}{\mu_{\Sigma}}\mid Z(t_0) = Z}\nonumber\\
	=&\sum_{t=t_0}^{t_0+T-1}\ex{\sum_{n=1}^N\sum_{m=1}^M \left(Q_n(t) - \widetilde{Q}_n^m(t) + \widetilde{Q}_n^m(t)\right) \beta_n^m(t)\lambda_m \mid Z}\nonumber\\
	&- \sum_{t=t_0}^{t_0+T-1} \ex{\sum_{n=1}^NQ_n(t)\frac{\epsilon \mu_n}{\mu_{\Sigma}}\mid Z}\nonumber\\
	=&\underbrace{\sum_{t=t_0}^{t_0+T-1}\ex{\sum_{n=1}^N\sum_{m=1}^M \left(Q_n(t) - \widetilde{Q}_n^m(t)\right) \beta_n^m(t)\lambda_m \mid Z}}_\text{$\mathcal{T}_1$}\nonumber \\
	+& \underbrace{\sum_{t=t_0}^{t_0+T-1}\ex{\sum_{n=1}^N\sum_{m=1}^M \widetilde{Q}_n^m(t)\beta_n^m(t)\lambda_m \mid Z}}_\text{$\mathcal{T}_2$} - \underbrace{\sum_{t=t_0}^{t_0+T-1} \ex{\sum_{n=1}^NQ_n(t)\frac{\epsilon \mu_n}{\mu_{\Sigma}}\mid Z}}_\text{$\mathcal{T}_3$}.\nonumber
\end{align}
 We are going to handle each term one by one. To upper bound $\mathcal{T}_1$, we use the following result on $X_n^m(t) = |Q_n(t) - \widetilde{Q}_n^m(t)|$.
\begin{lemma}
\label{lem:boundness}
	Under the condition given by Eq.~\eqref{eq:throughput_condition}, for any $t_0$ and $Z(t_0)$, there exists a finite $T_1$ independent of $\epsilon$ and a finite constant $L$ that is only a function of $p$ and $\mu_{\Sigma}$, such that for all $T \ge T_1$
	\begin{align*}
	 	\ex{ \sum_{t=t_0}^{t_0+T-1} X_m^n(t)\mid Z(t_0) = Z} \le L T
	 \end{align*} 
	 holds for all $m$ and $n$.
\end{lemma}
\begin{proof}
	See Appendix~\ref{proof:boundness}.
\end{proof}
By using Lemma~\ref{lem:boundness} with $T \ge T_1$, we have 
\begin{align}
\label{eq:T1}
	\mathcal{T}_1 &\le \lambda_{\Sigma}\sum_{t=t_0}^{t_0+T-1}\ex{\sum_{n=1}^N\sum_{m=1}^M \left|Q_n(t) - \widetilde{Q}_n^m(t)\right| \mid Z}\nonumber\\
	&\le \lambda_{\Sigma} MNL T.
\end{align}

For $\mathcal{T}_2$, we have 
\begin{align}
\label{eq:T2}
	\mathcal{T}_2 &\ep{a} \sum_{t=t_0}^{t_0+T-1}\ex{\sum_{n=1}^N\sum_{m=1}^M \widetilde{Q}_{\sigma_t(n)}^m(t)\Delta_n^m(t)\lambda_m \mid Z}\nonumber\\
	&\lep{b}0,
\end{align}
where (a) comes from the definition of dispatching preference vector $\Delta^m(t)$; (b) holds since dispatching preference is tilted and $\widetilde{Q}_{\sigma_t(1)}^m(t) \le \widetilde{Q}_{\sigma_t(2)}^m(t) \le \ldots \le \widetilde{Q}_{\sigma_t(N)}^m(t)$.

For $\mathcal{T}_3$, we have
\begin{align}
\label{eq:T3}
	\mathcal{T}_3 \ge \frac{\epsilon\mu_{min}}{\mu_{\Sigma}}\norm{\Q(t_0)}_1,
\end{align}
where $\mu_{min} = \min_n \mu_n$. 

Now, combining Eqs.~\eqref{eq:T1},~\eqref{eq:T2} and \eqref{eq:T3}, yields
\begin{align*}
	&\ex{\sum_{t=t_0}^{t_0+T-1} \inner{\Q(t)}{\A(t) -\s(t)} \mid Z(t_0) = Z}\nonumber\\
	\le& -\frac{\epsilon\mu_{min}}{\mu_{\Sigma}}\norm{\Q(t_0)}_1 + \lambda_{\Sigma} MNL T.
\end{align*}
Substituting the result above back into Eq.~\eqref{eq:driftW}, yields
\begin{align}
\label{eq:D_Q}
	D_Q(t_0)\le -2\frac{\epsilon\mu_{min}}{\mu_{\Sigma}}\norm{\Q(t_0)}_1 + 2\lambda_{\Sigma} MNL T + KT.
\end{align}
Now, we are ready to substitute Eq.~\eqref{eq:D_X} and Eq.~\eqref{eq:D_Q} back into Eq.~\eqref{eq:drift_all}. As a result, we have
\begin{align}
	D(Z(t_0)) \le & -2\frac{\epsilon\mu_{min}}{\mu_{\Sigma}}\norm{\Q(t_0)}_1 -p\sum_{m=1}^M\sum_{n=1}^NX_n^m(t_0)\nonumber\\
		& + 2\lambda_{\Sigma} MNL T + KT+ \lambda_{\Sigma} + \mu_{max}\nonumber\\
		\lep{a} & -\xi\left(\norm{\Q(t_0)}_1 + \sum_{m=1}^M\sum_{n=1}^N|Q_n(t_0) - \widetilde{Q}_n^m(t_0)|\right) + K_1,\nonumber
\end{align}
where in (a) $\xi = \min(2\frac{\epsilon\mu_{min}}{\mu_{\Sigma}},p)$ and $K_1 \triangleq 2\lambda_{\Sigma} MNL T + KT+ \lambda_{\Sigma} + \mu_{max}$. Pick any $\alpha>0$ and let 
\begin{align*}
	\mathcal{B} \triangleq \{Z \in \mathcal{Z}: \norm{\Q(t_0)}_1 + \sum_{m=1}^M\sum_{n=1}^N|Q_n(t_0) - \widetilde{Q}_n^m(t_0)| \le \frac{K_1+\alpha}{\xi} \}.
\end{align*}
Then, $\mathcal{B}$ is a finite subset. For any $Z\in \mathcal{B}^c$, $D(Z) \le -\alpha$, and for any $Z\in \mathcal{B}$, $D(Z) \le K_1$. By Foster-Lyapunov theorem, we have established positive recurrence. 

Having shown that the Markov chain $\{Z(t), t\ge 0\}$ is ergodic, we are left with the task of showing that all the moments are finite in steady-state. In order to do so, we use Lemma~\ref{lem:basis}. In particular, we choose the Lyapunov function as $V(Z^{(\epsilon)}) = \norm{\Q^{(\epsilon)}}$ and then verify the two conditions. In the following, the superscript $^{(\epsilon)}$ will be omitted for ease
of notations. To verify condition (C2), we have
\begin{align}
     |\Delta V(Z)| &= | \norm{\Q(t_0+T)} - \norm{\Q(t_0)} | \mathcal{I}(Z(t_0) = Z)\nonumber\\
          & \lep{a} \norm{\Q(t_0+T) - \Q(t_0)}\mathcal{I}(Z(t_0) = Z)\nonumber\\
          &\le \sum_{t=t_0}^{t_0+T-1}\norm{\Q(t+1) - \Q(t)}\mathcal{I}(Z(t_0) = Z)\nonumber\\
          &\le \sum_{t=t_0}^{t_0+T-1}\norm{\A(t)-\s(t)+\UU(t)}\mathcal{I}(Z(t_0) = Z)\nonumber\\
          &\lep{b} \sum_{t=t_0}^{t_0+T-1}\left(\norm{\A(t)} + 2\norm{\s(t)}\right)\mathcal{I}(Z(t_0) = Z),\label{eq:Qdriftbound}
  \end{align}
  where (a) holds since $| \norm{\mathbf{x}} - \norm{\mathbf{y}} | \le \norm{\mathbf{x} - \mathbf{y}}$ for each $\mathbf{x}$, $\mathbf{y}$ in $\mathbb{R}^N$. (b) follows from triangle inequality and the fact that $U_n(t) \le S_n(t)$ for all $t$ and $t$. Then, by our assumptions of light-tailed distributions for both total arrival and service processes, there exists a random variable $W$ such that $|\Delta V(X)| \prec W$ for all $t_0$ and all $X \in \mathcal{X}$, and $\ex{e^{\theta W}} = D$ is finite for some $\theta > 0$, which verifies (C2).

  For (C1), we have 
  \begin{align*}
    &\ex{ \Delta V(Z) \mid Z(t_0) = Z}\\
    = & \ex{\norm{\Q(t_0 + T)} - \norm{\Q(t_0)}  \mid Z(t_0) = Z}\\
    = & \ex{\sqrt{\norm{\Q(t_0 + T)}^2} - \sqrt{\norm{\Q(t_0)}^2} \mid Z(t_0) = Z }\\
    \lep{a} & \frac{1}{2\norm{\Q(t_0)}} \ex{\norm{\Q(t_0 + T)}^2 - \norm{\Q(t_0)}^2 \mid Z(t_0) = Z}\\
    \lep{b} & -\epsilon\frac{\mu_{min}}{\mu_{\Sigma}} + \frac{2\lambda_{\Sigma} MNL T + KT}{2\norm{\Q(t_0)}},
  \end{align*}
  where (a) follows from the fact that $f(x) = \sqrt{x}$ is concave; (b) comes from Eq. \eqref{eq:D_Q}. Thus, condition (C1) is valid and hence the proof of Theorem~\ref{thm:throughput} is complete.

\subsection{Proof of Theorem~\ref{thm:delay}}
\label{sec:proof_delay}
In order to prove the result, we need two intermediate results. One is called $\emph{state-space collapse}$ as stated in Proposition~\ref{lem:collapse}, which is the key ingredient for establishing heavy traffic delay optimality. Roughly speaking, it means that the multi-dimension space for the queue length vector reduces to one dimension in the sense that the deviation from the line (on which all the queue lengths are equal) is bounded by a constant, independent of $\epsilon$. Another  intermediate result is concerned with unused service. Based on these two intermediate results, we can prove heavy-traffic delay optimality. We omit the time reference $t_0$ for simplicity when necessary.

\begin{proposition}
    \label{lem:collapse}
    	Under the conditions in Theorem~\ref{thm:delay}, then we have that $\Qc$ is bounded in the sense that in steady state there exists finite constants $\{L_r, r \in \mathbb{N}\}$ independent of $\epsilon$ such that 
    	\begin{equation*}
    		\ex{\norm{\overline{\Q}_\perp^{(\epsilon)} } ^r} \le L_r
    	\end{equation*}
    	for all $\epsilon \in (0,\epsilon_0)$ and $r \in \mathbb{N}$, where $\Qc = \Q - \inner{\Q}{\mathbf{c}}\mathbf{c}$ is the perpendicular component of $\Q$ with respect to the line $\mathbf{c} = \frac{1}{\sqrt{N}}(1,1,\ldots,1)$.
\end{proposition}
\begin{proof}
	It suffices to show that $V_{\perp}(Z^{(\epsilon)}) \triangleq \norm{\Qc^{(\epsilon)}}$ satisfies the conditions (C1) and (C2) in Lemma~\ref{lem:basis}. Let us first consider conditions (C2). In particular, we have
	\begin{align}
    \label{eq:boundedQc}
            & |\Delta V_\perp(Z)| \nonumber\\
            = & | \norm{\Qc(t_0+T)} - \norm{\Qc(t_0)} | \mathcal{I}(Z(t_0) = Z)\nonumber \\
            \lep{a} & \norm{\Qc(t_0+T) - \Qc(t_0)}\mathcal{I}(Z(t_0) = Z)\nonumber\\
            = &\norm{\Q(t_0+T) - \Qp(t_0+T) - \Q(t_0) + \Qp(t_0)} \mathcal{I}(Z(t_0) = Z)\nonumber\\
            \lep{b} &\norm{\Q(t_0+T) - \Q(t_0)}  + \norm{\Qp(t_0+T) - \Qp(t_0)} \mathcal{I}(Z(t_0) = Z)\nonumber\\
            \lep{c} & 2\norm{\Q(t_0+T) - \Q(t_0)} \mathcal{I}(Z(t_0) = Z)\nonumber\\
            \lep{d} &2\sum_{t=t_0}^{t_0+T-1}\left(\norm{\A(t)} + 2\norm{\s(t)}\right)\mathcal{I}(Z(t_0) = Z)
    \end{align}
    where the inequality (a) follows from the fact that  $|\norm{{\bf x}} - \norm{{\bf y}}| \le \norm{{\bf x} - {\bf y}}$ holds for any ${\bf x}$, ${\bf y} \in \mathbb{R}^N$; inequality (b) follows from triangle inequality; (c) holds due to the non-expansive property of projection to a convex set; (d) follows from Eq.~\eqref{eq:Qdriftbound}. Then by our assumptions of light-tailed distributions for both total arrival and service processes, there exists a random variable $W$ such that $|\Delta V_{\perp}(X)| \prec W$ for all $t_0$ and all $X \in \mathcal{X}$, and $\ex{e^{\theta W}} = D$ is finite for some $\theta > 0$, which verifies (C2).

	Let us turn to condition (C1). By the proof of Lemma 3.6 in~\cite{zhou2017designing}, it suffices to establish the following result in order to verify (C1). That is, there exists $T>0$, $K_2\ge0$ and $\eta>0$ that are all independent of $\epsilon$, such that for all $t_0$ and $Z\in \mathcal{Z}$
	\begin{align}
	\label{eq:cond_C1}
		 \ex{\sum_{t=t_0}^{t_0+T-1} \inner{\Qc(t)}{\A(t) -\s(t)} \mid Z(t_0) = Z} \le -\eta \norm{\Qc} +K_2
	\end{align}
	holds for all $\epsilon \in (0,\epsilon_0)$. Note that
	\begin{align}
		&\ex{\sum_{t=t_0}^{t_0+T-1} \inner{\Qc(t)}{\A(t) -\s(t)} \mid Z(t_0) = Z}\nonumber\\
		\ep{a}& \sum_{t=t_0}^{t_0+T-1}\ex{\ex{\inner{\Qc(t)}{\A(t)}\mid Z(t)}\mid Z(t_0) = Z}\label{eq:Qc_A}\\
		 &- \sum_{t=t_0}^{t_0+T-1} \ex{\sum_n \mu_nQ_{\perp,n}(t) \mid Z(t_0) = Z},\label{eq:Qc_S}
	\end{align}
	where (a) follows from the tower property of conditional expectation and the fact that $A(t)$ is independent of $Z(t_0)$ given $Z(t)$. Moreover, $Q_{\perp,n}(t)$ denotes the $n$th component of the vector $\Qc(t)$. Now let us first focus on Eq.~\eqref{eq:Qc_A}.
	\begin{align}
		&\sum_{t=t_0}^{t_0+T-1}\ex{\ex{\inner{\Qc(t)}{\A(t)}\mid Z(t)}\mid Z(t_0) = Z}\nonumber\\
		=&\sum_{t=t_0}^{t_0+T-1}\ex{\sum_{n=1}^N Q_{\perp,n}(t)\sum_{m=1}^MP_n^m(t)\lambda_m \mid Z(t_0) = Z}\nonumber\\
		=&\sum_{t=t_0}^{t_0+T-1}\ex{\sum_{n=1}^N Q_{\perp,n}(t)\sum_{m=1}^M\left(\beta_n^m(t) + \frac{\mu_{n} }{\mu_{\Sigma}} \right)\lambda_m \mid Z(t_0) = Z}\nonumber\\
		=&\sum_{t=t_0}^{t_0+T-1}\ex{\sum_{n=1}^N Q_{\perp,n}(t)\sum_{m=1}^M\beta_n^m(t)\lambda_m\mid Z}\nonumber \\
		&+ \sum_{t=t_0}^{t_0+T-1}\ex{\sum_{n=1}^N Q_{\perp, n}(t)\sum_{m=1}^M\frac{\mu_{n} }{\mu_{\Sigma}}\left(\mu_{\Sigma} - \epsilon\right)p_m\mid Z}.\nonumber
	\end{align}
	Combining the result above with Eq.~\eqref{eq:Qc_S}, yields
	\begin{align}
	&\ex{\sum_{t=t_0}^{t_0+T-1} \inner{\Qc(t)}{\A(t) -\s(t)} \mid Z(t_0) = Z}\nonumber\\
		=&\sum_{t=t_0}^{t_0+T-1}\ex{\sum_{n=1}^N Q_{\perp, n}(t)\sum_{m=1}^M\beta_n^m(t)\lambda_m\mid Z}\label{eq:Qbeta}\\
		&+ \sum_{t=t_0}^{t_0+T-1} \ex{\sum_{n=1}^NQ_{\perp,n}(t)\frac{-\epsilon \mu_n}{\mu_{\Sigma}}\mid Z}.\label{eq:Qservice}
	\end{align}
	Note that by definition $Q_{\perp,n}(t) = Q_n(t) - Q_{\text{avg}}(t)$, in which $Q_{\text{avg}}(t)$ is the average queue length among $N$ queues at the beginning of time-slot $t$. Moreover, $Q_{\perp,n}(t)$ can be written as 
	\begin{align}
	\label{eq:decouple}
		Q_{\perp,n}(t) = Q_n(t) - \widetilde{Q}_n^m(t) + \widetilde{Q}_n^m(t) - \bar{Q}^m(t) + \bar{Q}^m(t) - Q_{\text{avg}}(t)
	\end{align}
	for all $m$ and $t$, in which $\bar{Q}^m(t) \triangleq \frac{1}{N}\sum_{n=1}^N \widetilde{Q}_n^m(t)$, i.e., the average queue length estimated by dispatcher $m$ at the beginning of time-slot $t$.
	By utilizing Eq.~\eqref{eq:decouple}, Eq.~\eqref{eq:Qbeta} can be written as 
	\begin{align}
		&\sum_{t=t_0}^{t_0+T-1}\ex{\sum_{n=1}^N Q_{\perp, n}(t)\sum_{m=1}^M\beta_n^m(t)\lambda_m\mid Z}\nonumber\\
		=&{\sum_{t=t_0}^{t_0+T-1}\ex{\sum_{n=1}^N \sum_{m=1}^M \left(\widetilde{Q}_n^m(t) - \bar{Q}^m(t)\right)\beta_n^m(t)\lambda_m\mid Z}}\label{eq:first}\\
		&+ \sum_{t=t_0}^{t_0+T-1}\ex{\sum_{n=1}^N \sum_{m=1}^M \left(Q_n(t)- \widetilde{Q}_n^m(t)\right)\beta_n^m(t)\lambda_m\mid Z}\label{eq:second}\\
   		&+ \sum_{t=t_0}^{t_0+T-1}\ex{\sum_{n=1}^N \sum_{m=1}^M \left(\bar{Q}^m(t)-Q_{\text{avg} }(t)\right)\beta_n^m(t)\lambda_m\mid Z}.\label{eq:third}
	\end{align}
	Our main task now is to upper bound each term above. Let us start with Eq.~\eqref{eq:first}. In particular, we can bound it by using the following result.
	\begin{lemma}
	\label{claim:neg_drift}
		There exist finite positive constants $\eta$ and $C$ such that
		\begin{align*}
			{\sum_{t=t_0}^{t_0+T-1}\ex{\sum_{n=1}^N \sum_{m=1}^M \left(\widetilde{Q}_n^m(t) - \bar{Q}^m(t)\right)\beta_n^m(t)\lambda_m\mid Z}}\le -\eta\norm{\Qc(t_0)} + C
		\end{align*}
		holds for all $T \ge 3$, in which $\eta = \frac{\lambda_{\Sigma}\delta p^2}{\sqrt{N}}$ and $C = 3(\mu_{\Sigma})^2p^2$.
	\end{lemma}
	\begin{proof}
		See Appendix \ref{proof:neg_drift}
	\end{proof}
	For Eqs.~\eqref{eq:second} and \eqref{eq:third}, we can bound both of them by using the result in Lemma~\ref{lem:boundness}, respectively. In particular, for Eq.~\eqref{eq:second}, we have 
	\begin{align}
		&\sum_{t=t_0}^{t_0+T-1}\ex{\sum_{n=1}^N \sum_{m=1}^M \left(Q_n(t)- \widetilde{Q}_n^m(t)\right)\beta_n^m(t)\lambda_m\mid Z}\nonumber\\
		&\le \lambda_{\Sigma}\sum_{t=t_0}^{t_0+T-1}\ex{\sum_{n=1}^N\sum_{m=1}^M \left|Q_n(t) - \widetilde{Q}_n^m(t)\right| \mid Z}\nonumber\\
	&\le \mu_{\Sigma} MNL T.
	\end{align}
	For Eq.~\eqref{eq:third}, we have 
	\begin{align}
		&\sum_{t=t_0}^{t_0+T-1}\ex{\sum_{n=1}^N \sum_{m=1}^M \left(\bar{Q}^m(t)-Q_{\text{avg} }(t)\right)\beta_n^m(t)\lambda_m\mid Z}\nonumber\\
		=&\sum_{t=t_0}^{t_0+T-1}\ex{\sum_{n=1}^N \sum_{m=1}^M \left(\frac{1}{N}\sum_{n=1}^N\left(\widetilde{Q}_n^m(t) - Q_n(t) \right)\right)\beta_n^m(t)\lambda_m\mid Z}\nonumber\\
		\le &\lambda_{\Sigma}\sum_{t=t_0}^{t_0+T-1}\ex{\sum_{n=1}^N \sum_{m=1}^M \frac{1}{N}\sum_{n=1}^N\left|\widetilde{Q}_n^m(t) - Q_n(t) \right|\mid Z}\nonumber\\
		\le &\mu_{\Sigma} MNL T.
	\end{align}
	We have obtained bounds for Eqs.~\eqref{eq:first},~\eqref{eq:second} and~\eqref{eq:third}. Let us turn to focus on Eq.~\eqref{eq:Qservice}, which can be upper bounded by the following result.
	\begin{lemma}
	\label{lem:Qservice}
	For any $t_0$ and $Z$,
	\begin{align*}
		\sum_{t=t_0}^{t_0+T-1} \ex{\sum_{n=1}^NQ_{\perp,n}(t)\frac{-\epsilon \mu_n}{\mu_{\Sigma}}\mid Z(t_0) = Z} \le \epsilon\sqrt{N}T \norm{\Q(t_0)} + K_3,
	\end{align*}
	\end{lemma}
	where $k_3$ is a finite constant independent of $\epsilon$.
	\begin{proof}
		See Appendix~\ref{proof:Qservice}
	\end{proof}

	Now, we are ready to bound the left-hand-side of Eq.~\eqref{eq:cond_C1} by using the bounds for both Eq.~\eqref{eq:Qbeta} and Eq.~\eqref{eq:Qservice}. In particular, we have
	\begin{align}
		&\ex{\sum_{t=t_0}^{t_0+T-1} \inner{\Qc(t)}{\A(t) -\s(t)} \mid Z(t_0) = Z}\nonumber\\
		\le & - \frac{\lambda_{\Sigma}\delta p^2}{\sqrt{N}}\norm{\Qc(t_0)} + C + 2\mu_{\Sigma} MNL T +  \epsilon\sqrt{N}T \norm{\Q(t_0)} + K_3\nonumber\\
		\ep{a} &\left(T\epsilon - \frac{\lambda_{\Sigma}\delta p^2}{{N}}\right) \sqrt{N}\norm{\Qc(t_0)} + K_2 \nonumber\\
		\le & -\frac{\mu_{\Sigma}\delta p^2}{2\sqrt{N} }\norm{\Qc(t_0)} + K_2, \quad \forall \epsilon < \frac{\mu_{\Sigma}\delta p^2}{2NT+2\delta p^2}
	\end{align}
	where (a) follows from $K_2 = C + 2\mu_{\Sigma} MNL T + K_3$, which is independent of $\epsilon$. Hence, this verifies condition (C1) with $\eta =\frac{\mu_{\Sigma}\delta p^2}{2\sqrt{N} }$, which is also independent of $\epsilon$. Combined with condition (C2), we have finished the proof of Proposition~\ref{lem:collapse}.
\end{proof}

Having proved the state-space collapse result, we turn to prove another intermediate result regarding unused service, as stated in the following lemma. In words, this lemma says that in heavy traffic unused service tends to be zero.
\begin{lemma}
	\label{lem:unused_service}
	Under any LED policy, we have 
	\begin{align*}
		\lim_{\epsilon \downarrow 0} \ex{\norm{\overline{\UU}^{(\epsilon)}}^2} = 0.
	\end{align*}
	
\end{lemma}
\begin{proof}
First, we would like to show that under any LED policy,
  \begin{equation}
  \label{eq:unused_one_norm}
    \ex{\norm{\overline{\UU}^{(\epsilon)}}_1} = \epsilon.
  \end{equation}
  To see this, we consider the Lyapunov function $W_1(Z(t)) = \norm{\Q(t)}_1$. Since LED is throughput optimal with all the moments being finite, we have that the mean drift of $W_1(Z(t))$ in steady-state is zero. Then, we have
	\begin{equation*}
	\begin{split}
	  0 = \ex{\norm{\A^{(\epsilon)}}_1 - \norm{\s}_1 + \norm{\overline{\UU}^{(\epsilon)} }_1},
	\end{split}  
	\end{equation*}
	which directly implies the result in Eq.~\eqref{eq:unused_one_norm}.

	Now let us fix $n \in \mathcal{N}$, we have for any $t \ge 0$ and constant $S^{\prime}$
\begin{align*}
	{U}_n^2(t) &\le U_n(t)S_n(t)\\
	& = U_n(t)S_n(t) \mathcal{I}\left(S_n(t) \le S^{\prime} \right) + U_n(t)S_n(t) \mathcal{I}\left(S_n(t) > S^{\prime} \right)\\
	& \le U_n(t) S^{\prime} + S_n^2(t)\mathcal{I}\left(S_n(t) > S^{\prime} \right).
\end{align*}
In steady state, we have 
\begin{align*}
	\ex{\overline{U}_n^2} &\le \ex{\overline{U}_n}S^{\prime} + \ex{S_n^2(\infty)\mathcal{I}\left(S_n(\infty) > S^{\prime} \right)}\\
	&\lep{a} \epsilon S^{\prime} + \ex{S_n^2(0)\mathcal{I}\left(S_n(0) > S^{\prime} \right)}\\
	&\lep{b} \epsilon S^{\prime} + \beta,
\end{align*}
where (a) follows from the fact that $\ex{\norms{\overline{\UU}^{(\epsilon)}}_1} = \epsilon$ and service process is \emph{i.i.d.}; in (b), we choose $S^{\prime}$ such that $\ex{S_n^2(0)\mathcal{I}\left(S_n(0) > S^{\prime} \right)} \le \beta$, which is possible by the exponential decay rate of $S_n(0)$ under the light-tailed assumption. Thus, we have 
\begin{align*}
	\lim_{\epsilon \downarrow 0} \ex{\overline{U}_n^2} \le \beta,
\end{align*}
for any $\beta > 0$. Hence, we have $\lim_{\epsilon \downarrow 0} \ex{\overline{U}_n^2} = 0$ for each $n$, which directly implies our result.
\end{proof}

Now, we are prepared to show that under the conditions in Theorem~\ref{thm:delay}, the system achieves optimal delay in heavy traffic. More specifically, by Lemma 3 in~\cite{zhou2018flexible}, we need only to verify the following condition.
\begin{align}
		\label{eq:necessary}
			\lim_{\epsilon \downarrow 0}\ex{\big\lVert\overline{\Q}^{(\epsilon)}(t+1) \big\rVert_1 \big\lVert\overline{\UU}^{(\epsilon)}(t) \big\rVert_1} = 0.
\end{align}
Let us define $\overline{B}^{(\epsilon)} \triangleq\ex{\norm{\overline{\Q}^{(\epsilon)}(t+1)}_1 \norm{\overline{\UU}^{(\epsilon)} (t)}_1}$. We can bound it as follows.
\begin{equation*}
    \begin{split}
      \overline{B}^{(\epsilon)} & \ep{a} N\ex{\inner{\overline{\UU}^{(\epsilon)}(t)} {-\overline{\Q}^{(\epsilon)}_{\perp}(t+1)}}\\
      & \lep{b} N \sqrt{\ex{\norm{\overline{\UU}^{(\epsilon)}}^2} \ex{\norm{\overline{\Q}^{(\epsilon)}_{\perp}(t+1)}^2}}\\
      & \ep{c} N \sqrt{\ex{\norm{\overline{\UU}^{(\epsilon)}}^2} \ex{\norm{\overline{\Q}^{(\epsilon)}_{\perp}(t)}^2}}\\
      & \lep{d} N \sqrt{\ex{\norm{\overline{\UU}^{(\epsilon)}}^2} L_2},
    \end{split} 
    \end{equation*}
    where the equality (a) comes from the property $Q_n^{(\epsilon)}(t+1)U_n^{(\epsilon)}(t) = 0$ for all $n\in \mathcal{N}$ and all $t\ge0$ and the definition of $\Qc$; the inequality (b) holds due to Cauchy-Schwartz inequality; the equality (c) is true since the distributions of $\overline{\Q}^{(\epsilon)}_{\perp}(t+1)$ and $\overline{\Q}^{(\epsilon)}_{\perp}(t)$ are the same in steady state; (d) follow from the state-space collapse result in Proposition~\ref{lem:collapse}. Finally, by Lemma~\ref{lem:unused_service} and the fact that $L_2$ is independent of $\epsilon$, we have $\lim_{\epsilon \to 0} \overline{B}^{(\epsilon)} = 0$, which finishes our proof.

\section{Conclusion}
We have introduced the Local-Estimation-Driven (LED) framework for load balancing policies in possibly heterogeneous systems with multiple dispatchers. Under this framework, each dispatcher keeps local and possibly outdated estimates of the queue lengths for all the servers, and makes its dispatching decision only based on these local estimates. Communication between dispatchers and servers is only used to update the local estimates. We have established sufficient conditions for LED policies to achieve both throughput optimality and delay optimality in heavy traffic. These sufficient conditions not only establish delay optimality for many previous local-memory based policies, but enable us to tailor the design of new delay optimal policies based on different application requirements. The heavy-traffic delay optimality of LED policies also resolves a recent open problem on the development of load balancing schemes that have only access to delayed information.

In future work, it will be interesting to investigate LED framework in other asymptotic regimes, e.g., the large-system regime and the many-server heavy-traffic regime.

\bibliographystyle{plain}
\bibliography{ref} 

\section*{Appendix}
\appendix
\section{Proof of Lemma~\ref{lem:boundness}}
\label{proof:boundness}
First, we show that the Markov chain $\{X_n^m(t), t\ge0\}$ is ergodic. It is irreducible and aperiodic since any state can reach the initial state $X_n^m(0) = 0$ via update and the initial state can also reach to itself. It is positive recurrent since the expected return time to state $0$ is finite under the condition in Eq.~\eqref{eq:throughput_condition}. By the ergodicity of Markov chain, we have that the limiting distribution $\pi = (\pi_0, \pi_1, \ldots )$ exists and it is also the unique stationary distribution. Here, 
\begin{align*}
      \pi_k = \lim_{n\to\infty} \mathbb{P}_{vk}^n \quad\quad \forall v
\end{align*}
where $\mathbb{P}_{vk}^n$ is the probability of being in state $k$ in $n$ steps, given we are in state $v$ now. Since the limiting distribution is independent of the initial state, we can pick any state to start with. In particular, we pick the initial state as $X_n^m(0) = 0$. By Eq.~\eqref{eq:drift_X}, we have for any $t$
\begin{align}
\label{eq:relation}
	\ex{X_n^m(t+1)} \le (1-p) \ex{X_n^m(t)} + \lambda_{\Sigma} + \mu_{max}.
\end{align}
We conduct an inductive process to establish that $\ex{X_n^m(t)} \le \frac{\lambda_{\Sigma} + \mu_{max}}{p}$ for all $t$. First, the basis is true since $X_n^m(0) = 0$. We assume $\ex{X_n^m(t_0)} \le \frac{\lambda_{\Sigma} + \mu_{max}}{p}$ holds for $t = t_0$, then by Eq.~\eqref{eq:relation}, we have $\ex{X_n^m(t_0+1)} \le \frac{\lambda_{\Sigma} + \mu_{max}}{p}$ also holds and hence completing the proof.
In particular, $\ex{X_n^m(t)} \le \frac{\lambda_{\Sigma} + \mu_{max}}{p} \le \frac{2\mu_{\Sigma}}{p}$ for all $t$. This directly means that the stationary distribution $\pi$ of $\{X_n^m(t), t\ge0\}$ has a finite bound on the mean (independent of $\epsilon$). Similarly, by applying the same step as in Eq.~\eqref{eq:drift_X} combined with the inductive argument, we can obtain that all the moments of the stationary distribution are bounded, because all the moments of the total arrival and each service are bounded (independent of $\epsilon$) by our light-tailed assumption. 

Let $f_{T} \triangleq \frac{1}{T}\left(\sum_{t=t_0}^{t_0+T-1}X_n^m(t)\mid Z(t_0) = Z\right)$. By ergodicity (i.e., time-average = ensemble-average), we have that for any starting point $Z(t_0) = Z$, with probability $1$ such that
    \begin{align*}
      \lim_{T\to\infty}f_{T} = \sum_{k=0}^{\infty}k\pi_k \le \frac{2\mu_{\Sigma}}{p}.
    \end{align*}
As a result, we can find a finite $T_1$ (independent of $\epsilon$ since all the moments of $\pi$ are bounded with independence of $\epsilon$) such that for all $T \ge T_1$, $f_{T} \le \frac{4\mu_{\Sigma}}{p} \triangleq L$ with probability $1$. Therefore, if $T \ge T_1$, 
 \begin{align*}
  \ex{ \sum_{t=t_0}^{t_0+T-1} X_n^m(t)\mid Z(t_0) = Z} \le L T,
\end{align*}
which completes the proof.

\section{Proof of Lemma~\ref{claim:neg_drift}}
\label{proof:neg_drift}
Consider the left-hand-side (LHS) of the inequality in Lemma~\ref{claim:neg_drift}.
\begin{align}
	LHS\lep{a}&\sum_{t=t_0}^{t_0+T-1}\sum_{m=1}^M \lambda_m\delta\ex{ \widetilde{Q}^m_{min}(t) - \widetilde{Q}^m_{max}(t)\mid Z(t_0) = Z}\nonumber\\
	\lep{b}&\sum_{m=1}^M \lambda_m \delta\ex{ \widetilde{Q}^m_{min}(t_0+1) - \widetilde{Q}^m_{max}(t_0+1)\mid Z(t_0) = Z}\nonumber\\
	&+ \sum_{m=1}^M \lambda_m \delta\ex{ \widetilde{Q}^m_{min}(t_0+2) - \widetilde{Q}^m_{max}(t_0+2)\mid Z(t_0) = Z},\nonumber
\end{align}
where (a) follows from the definition of $\delta$-tilted dispatching strategy and the fact $\widetilde{Q}^m_{min}(t) \triangleq \widetilde{Q}_{\sigma_t(1)}^m(t) \le \widetilde{Q}_{\sigma_t(2)}^m(t) \le \ldots \le \widetilde{Q}_{\sigma_t(N)}^m(t) \triangleq \widetilde{Q}^m_{max}(t)$; (b) holds since all the terms in the summation are non-positive and $T \ge 3$. 

Now, we define random variable $\mathcal{I}_{min}^m(t)$ as the indicator function which obtains $1$ if the server with the minimal true queue length at the end of time-slot $t$ is updated by dispatcher $m$. Similarly, $\mathcal{I}_{max}^m(t)$ is defined as an indicator function which obtains $1$ if the server with the maximal true queue length  at the end of time-slot $t$ is updated by dispatcher $m$. In the following, we consider the event that $\mathcal{I}_{max}^m(t_0) = 1$ and $\mathcal{I}_{min}^m(t_0+1) = 1$, i.e., at the end of time-slot $t_0$ the server with the maximal actual queue is updated and at the end of time-slot $t_0+1$ the server with the minimal actual queue is updated. Then, by law of expectation, we have 
\begin{equation}
\label{eq:LHS}
\begin{aligned}
	&LHS\le \sum_{m=1}^M \lambda_m \delta p^2\mathbb{E}\left[\widetilde{Q}^m_{min}(t_0+1) - \widetilde{Q}^m_{max}(t_0+1)\right.\\
	&\left.\widetilde{Q}^m_{min}(t_0+2) - \widetilde{Q}^m_{max}(t_0+2)\mid Z, \mathcal{I}_{max}^m(t_0) = 1, \mathcal{I}_{min}^m(t_0+1) = 1\right]\\
	&\lep{a}\sum_{m=1}^M \lambda_m \delta p^2\mathbb{E}\left[\widetilde{Q}^m_{min}(t_0+1) - {Q}_{max}(t_0+1)\right.\\
	&\left.{Q}_{min}(t_0+2) - \widetilde{Q}^m_{max}(t_0+2)\mid Z, \mathcal{I}_{max}^m(t_0) = 1, \mathcal{I}_{min}^m(t_0+1) = 1\right]\\
	&\lep{b} \sum_{m=1}^M \lambda_m \delta p^2\mathbb{E}\left[ - {Q}_{max}(t_0+1) + {Q}_{min}(t_0+2)\right.\\
	&\left.\mid Z(t_0) = Z, \mathcal{I}_{max}^m(t_0) = 1, \mathcal{I}_{min}^m(t_0+1) = 1\right],
\end{aligned}	
\end{equation}
where (a) follows from the fact that with $\mathcal{I}_{max}^m(t_0) = 1$, $\widetilde{Q}^m_{max}(t_0+1) \ge {Q}_{max}(t_0+1)$ and with $\mathcal{I}_{min}^m(t_0+1) = 1$, $\widetilde{Q}^m_{min}(t_0+2) \le {Q}_{min}(t_0+2)$; (b) holds since $\widetilde{Q}^m_{min}(t_0+1) \le \widetilde{Q}^m_{max}(t_0+1) \le \widetilde{Q}^m_{max}(t_0+2)$. This is because under any LED policy the decrease of local estimation can only happen when the queue is updated. 

In order to relate the queue lengths at time-slots $t_0+1$ and $t_0+2$ to the queue length at $t_0$, we use the following result.
\begin{claim}
\label{clm:bound}
	For any $t_0$, we have 
  \begin{enumerate}
    \item $\ex{{Q}_{min}(t_0+1) \mid Z(t_0) = Z} \le Q_{min}(t_0) + M_1$
    \item $\ex{{Q}_{max}(t_0+1) \mid Z(t_0) = Z} \ge Q_{max}(t_0) - M_1$
  \end{enumerate}
  where $M_1 = \mu_{\Sigma}$.
\end{claim}
\begin{proof}
  Let us start with the first result. Suppose that the server $i$ is the server with the shortest queue length at time-slot $t_0$. We have 
  \begin{align}
  \label{eq39}
    \ex{{Q}_{i}(t_0+1) \mid Z(t_0) = Z} &= \ex{Q_i(t_0) + A_i(t_0) - S_i(t_0) +U_i(t) \mid Z(t_0)}\nonumber\\
    &\le Q_{i}(t_0) + \max(\lambda_{\Sigma}, \mu_i)\nonumber\\
    &\le Q_{i}(t_0) + \mu_{\Sigma}\nonumber\\
    & = Q_{i}(t_0) + M_1.
  \end{align}
  If at time-slot $t_0+1$, the same server $i$ is still the one with the shortest queue length, then we are done. If not, suppose that some other server $j$ is the one with the shortest queue length at time-slot $t_0+1$. Now, we assume that $\ex{{Q}_{j}(t_0+1) \mid Z(t_0) = Z} > Q_{i}(t_0) + M_1$, and we hope to arrive at a contradiction.

  First, since ${Q}_{j}(t_0+1) \le {Q}_{i}(t_0+1)$, $\ex{{Q}_{j}(t_0+1) \mid Z(t_0) = Z} \le \ex{{Q}_{i}(t_0+1) \mid Z(t_0) = Z}$. Combined with our assumption, we get 
  \begin{align}
  \label{eq40}
    Q_{i}(t_0) + M < \ex{{Q}_{j}(t_0+1) \mid Z(t_0) = Z} \le \ex{{Q}_{i}(t_0+1) \mid Z(t_0) = Z}
  \end{align}
Thus, we can see Eq. \eqref{eq40} contradicts with Eq. \eqref{eq39}. Hence, 
\begin{align*}
	\ex{{Q}_{j}(t_0+1) \mid Z(t_0) = Z} \le Q_{i}(t_0) + M_1,
\end{align*}
which finishes the proof of the first result. Same argument can be applied to prove the second result.
\end{proof}

Now, we are ready to bound Eq.~\eqref{eq:LHS}. First, 
\begin{align}
	&LHS \lep{a} \sum_{m=1}^M \lambda_m \delta p^2\mathbb{E}\left[ - {Q}_{max}(t_0+1) \mid Z(t_0) = Z\right] + \label{eq:32}\\
	&\sum_{m=1}^M \lambda_m \delta p^2\mathbb{E}\left[{Q}_{min}(t_0+2)\mid Z(t_0) = Z,\phi(t_0) \right],\label{eq:33}
\end{align}
where $\phi(t_0) \triangleq \mathcal{I}_{max}^m(t_0) = 1, \mathcal{I}_{min}^m(t_0+1) = 1$. The inequality follows from the fact that given $Z(t_0)$, $Q_{max}(t_0+1)$ is independent of $\mathcal{I}_{max}^m(t_0)$ and $\mathcal{I}_{min}^m(t_0+1)$. By using the bound in Claim~\ref{clm:bound}, we have an upper bound for the term in Eq.~\eqref{eq:32}
\begin{align}
\label{eq:34}
	&\sum_{m=1}^M \lambda_m \delta p^2\mathbb{E}\left[ - {Q}_{max}(t_0+1) \mid Z(t_0) = Z\right]\nonumber\\
	 &\le \sum_{m=1}^M \lambda_m \delta p^2(-Q_{max}(t_0)+M_1).
\end{align}
For the term in Eq.~\eqref{eq:33}, we have 
\begin{align*}
	&\sum_{m=1}^M \lambda_m \delta p^2\mathbb{E}\left[{Q}_{min}(t_0+2)\mid Z(t_0) = Z, \phi(t_0)\right]\\
		\ep{a} &\sum_{m=1}^M \lambda_m \delta p^2\mathbb{E}\left[ \ex{{Q}_{min}(t_0+2) \mid Z(t_0+1)}\mid Z(t_0) = Z,\right. \\
		&\left.\mathcal{I}_{max}^m(t_0) = 1, \mathcal{I}_{min}^m(t_0+1) = 1\right]\\
\end{align*}
where (a) follows from the tower property of conditional expectation and the fact that given $Z(t_0+1)$, ${Q}_{min}(t_0+2)$ is independent of $\phi(t_0)$. Then, it can be upper bounded as follows.
\begin{equation*}
\label{eq:35}
	\begin{aligned}
		&\sum_{m=1}^M \lambda_m \delta p^2\mathbb{E}\left[{Q}_{min}(t_0+2)\mid Z(t_0) = Z, \phi(t_0)\right]\\
		\lep{a}&\sum_{m=1}^M \lambda_m \delta p^2\mathbb{E}\left[Q_{min}(t_0+1) + M_1\mid Z(t_0) = Z,\right. \\
		&\left.\mathcal{I}_{max}^m(t_0) = 1, \mathcal{I}_{min}^m(t_0+1) = 1\right]\\
		\ep{b} &\sum_{m=1}^M \lambda_m \delta p^2\mathbb{E}\left[Q_{min}(t_0+1) + M_1\mid Z(t_0) = Z\right]\\
		\lep{c}&\sum_{m=1}^M \lambda_m \delta p^2(Q_{min}(t_0) + 2M_1),
	\end{aligned}
\end{equation*}
 where (a) comes from the bound in Claim~\ref{clm:bound}; (b) holds since given $Z(t_0)$, $Q_{min}(t_0+1)$ is independent of the event $\mathcal{I}_{max}^m(t_0) = 1, \mathcal{I}_{min}^m(t_0+1) = 1$; (c) holds by the bound in Claim~\ref{clm:bound} again.

Thus, combining the bounds for Eqs.~\eqref{eq:32} and~\eqref{eq:33}, yields
\begin{align}
	LHS &\le \sum_{m=1}^M \lambda_m \delta p^2 (Q_{min}(t_0) -Q_{max}(t_0) + 3M_1 )\nonumber\\
	&=\lambda_{\Sigma}\delta p^2(Q_{min}(t_0) -Q_{max}(t_0) + 3M_1\lambda_{\Sigma}\delta p^2\nonumber\\
	&\le -\frac{\lambda_{\Sigma}\delta p^2}{\sqrt{N}}\norm{\Qc(t_0)} + 3(\mu_{\Sigma})^2p^2,
\end{align}
in which the last inequality follows from the fact that $\norm{\Qc(t_0)} \le \sqrt{N}\left(Q_{max}(t_0)-Q_{min}(t_0)\right)$ and $M_1 = \mu_{\Sigma}$ with $\delta \le 1$. Hence, the proof of Lemma~\ref{claim:neg_drift} is complete.

\section{Proof of Lemma~\ref{lem:Qservice} }
\label{proof:Qservice}
First, note that by Eq.~\eqref{eq:boundedQc}, we have 
\begin{align}
	&\ex{\norm{\Qc(t + 1)} \mid Z(t) = Z}\nonumber\\
	\le &\norm{\Qc(t)} + 2\ex{\norm{\A(t)} + 2\norm{\s(t)} \mid Z}\nonumber\\
	\le &\norm{\Qc(t)}+2\ex{\sqrt{NA_{\Sigma}^2(t)} + 2\sqrt{NS_{\Sigma}^2(t)} \mid Z}\nonumber\\
	\lep{a} & \norm{\Qc(t)}+ 2\sqrt{N(\sigma_{\Sigma}^2 + \mu_{\Sigma}^2)} + 4\sqrt{N(\nu_{\Sigma}^2+\mu_{\Sigma}^2)}\nonumber\\
	\lep{b} & \norm{\Qc(t)} + C_2\label{eq:qcdiff},
\end{align}
where (a) follows from Jensen's inequality for concave function; in (b) $C_2$ is a finite constant independent of $\epsilon$, which holds by our light-tailed assumption.

Now, by using the result above, we have 
\begin{align}
	&\sum_{t=t_0}^{t_0+T-1} \ex{\sum_{n=1}^NQ_{\perp,n}(t)\frac{-\epsilon \mu_n}{\mu_{\Sigma}}\mid Z(t_0) = Z}\nonumber \\
	\lep{a} &\sum_{t=t_0}^{t_0+T-1}\ex{\epsilon \sqrt{N}\norm{\Qc(t)} \mid Z(t_0) = Z}\nonumber\\
	\le &\epsilon\sqrt{N}\left(\sum_{t=t_0}^{t_0+T-1} \norm{\Qc(t_0)} + (t-t_0)C_2\right)\nonumber\\
	\le &T\epsilon\sqrt{N}\norm{\Qc(t_0)} + T^2C_2\nonumber,
\end{align}
where (a) is true since $\norm{\mathbf{x}}_1\le \sqrt{N}\norm{\mathbf{x}}$ for any $\mathbf{x} \in \mathbb{R}^N$. Hence, the proof is complete.

\end{document}